\documentclass[screen,authorversion,nonacm]{acmart}

\AtBeginDocument{%
  \providecommand\BibTeX{{%
    \normalfont B\kern-0.5em{\scshape i\kern-0.25em b}\kern-0.8em\TeX}}}

\usepackage{mathtools,dsfont}
\usepackage{enumitem}
\usepackage{mathrsfs}

\newtheorem{theorem}{Theorem}[section]
\newtheorem{fact}[theorem]{Fact}
\newtheorem{claim}{Claim}[section]
\theoremstyle{acmdefinition}
\newtheorem{remark}[theorem]{Remark}
\newtheorem{prop}[theorem]{Proposition}

\newcommand{\bool}{{\{0,1\}}}
\newcommand{\sbool}{\{-1,1\}}

\DeclareMathOperator{\ent}{\ensuremath \mathbb{H}}
\DeclareMathOperator{\bor}{\ensuremath \mathsf{OR}}
\DeclareMathOperator{\band}{\ensuremath \mathsf{AND}}
\DeclareMathOperator{\E}{\ensuremath \mathbb{E}}

\newcommand{\fc}[2]{\ensuremath {\widehat{#1}(#2)}}
\newcommand{\fcsq}[2]{\ensuremath {{\widehat{#1}}(#2)^2}}
\newcommand{\reals}{\ensuremath {\mathbb{R}}}

\newcommand{\s}{\ensuremath {\mathsf{s}}}
\newcommand{\as}{\ensuremath {\mathsf{as}}}
\renewcommand{\deg}{\ensuremath {\mathsf{deg}}}
\newcommand{\certificate}{\ensuremath {\mathsf{C}}}
\newcommand{\ac}{\ensuremath {\mathsf{aC}}}

\newcommand{\aUC}{\ensuremath {\mathsf{aUC}}}
\newcommand{\UC}{\ensuremath {\mathsf{UC}}}
\newcommand{\dnf}{\ensuremath {\mathsf{DNF}}}
\newcommand{\Inf}{\ensuremath {\mathsf{Inf}}}
\newcommand{\Ent}{\ensuremath {\mathsf{Ent}}}
\newcommand{\supp}{\ensuremath {\mathsf{supp}}}
\newcommand{\Var}{\ensuremath {\mathsf{Var}}}
\newcommand{\pmset}[1]{\{-1,1\}^{#1}} 

\begin{document}

\title{Improved bounds on Fourier entropy and Min-entropy}

\author{Srinivasan Arunachalam}
\affiliation{%
  \institution{IBM Research}
  \city{New York}
  \country{USA}
}
\email{Srinivasan.Arunachalam@ibm.com}

\author{Sourav Chakraborty}
\affiliation{%
  \institution{Indian Statistical Institute}
  \city{Kolkata}
  \country{India}
}
\email{sourav@isical.ac.in}

\author{Michal Kouck\'{y}}
\affiliation{%
  \institution{Computer Science Institute of Charles University}
  \city{Prague}
  \country{Czech Republic}
}
\email{koucky@iuuk.mff.cuni.cz}

\author{Nitin Saurabh}
\affiliation{%
  \institution{Technion - IIT}
  \city{Haifa}
  \country{Israel}
}
\email{nitinsau@cs.technion.ac.il}

\author{Ronald de Wolf}
\affiliation{%
  \institution{QuSoft, CWI and University of Amsterdam}
  \city{Amsterdam}
  \country{the Netherlands}
}
\email{rdewolf@cwi.nl}

\renewcommand{\shortauthors}{Arunachalam et al.}

\begin{abstract}
 Given a Boolean function $f\colon\pmset{n}\rightarrow \pmset{}$,
  define the \emph{Fourier distribution} to be the distribution on
  subsets of $[n]$, where each $S\subseteq [n]$ is sampled with
  probability $\fcsq{f}{S}$.
  The Fourier Entropy-Influence (FEI) conjecture of
  Friedgut and Kalai~\cite{FK96} seeks to relate
  two fundamental measures associated with the Fourier distribution:
  does there exist a universal constant $C>0$ such that
  $\ent(\hat{f}^2)\leq C\cdot \Inf(f)$,
  where $\ent(\hat{f}^2)$ is the Shannon entropy of
  the Fourier distribution of $f$ and
  $\Inf(f)$ is the total influence of $f$?

  In this paper we present three new contributions towards the FEI conjecture:
  \begin{enumerate}
  \item Our first contribution shows that
    $\ent(\hat{f}^2)\leq 2\cdot \aUC^\oplus(f)$,
    where $\aUC^\oplus(f)$ is the average unambiguous
    parity-certificate complexity of $f$.
    This improves upon several bounds shown by
    Chakraborty et al.~\cite{CKLS16}.
    We further improve this bound for unambiguous $\dnf$s.
    We also discuss how our work makes Mansour's conjecture
    for $\dnf$s a natural next step towards resolution of the FEI conjecture. 

  \item We next consider the weaker
    \emph{Fourier Min-entropy-Influence} (FMEI) conjecture posed by
    O'Donnell and others~\cite{OWZ11,ODonnell-book} which asks if
    $\ent_{\infty}(\hat{f}^2) \leq C\cdot \Inf(f)$, where
    $\ent_{\infty}(\hat{f}^2)$ is the min-entropy of the Fourier distribution.
    We show $\ent_{\infty}(\hat{f}^2)\leq  2\cdot\certificate_{\min}^\oplus(f)$,
    where $\certificate_{\min}^\oplus(f)$ is the minimum
    parity-certificate complexity of $f$. We also show that
    for all $\varepsilon\geq 0$, we have
    $\ent_{\infty}(\hat{f}^2)\leq 2\log (\|\widehat{f}\|_{1,\varepsilon}/(1-\varepsilon))$, where $\|\widehat{f}\|_{1,\varepsilon}$ is
    the approximate spectral norm of $f$. As a corollary,
    we verify the FMEI conjecture for the class of read-$k$ $\dnf$s
    (for constant~$k$). 

  \item Our third contribution is to better understand implications
    of the FEI conjecture for the structure of polynomials that
    $1/3$-approximate a Boolean function on the Boolean cube.
    We pose a conjecture$\colon$ no \emph{flat polynomial}
    (whose non-zero Fourier coefficients have the same magnitude)
    of degree $d$ and sparsity $2^{\omega(d)}$ can $1/3$-approximate
    a Boolean function. This conjecture is known to be true assuming FEI
    and we prove the conjecture unconditionally
    (i.e., without assuming the FEI conjecture) for a class of polynomials.
    We  discuss an intriguing connection between our conjecture and
    the constant for the Bohnenblust-Hille inequality,
    which has been extensively studied in functional analysis.
  \end{enumerate}
\end{abstract}



\maketitle

\section{Introduction}
\label{sec:intro}
Boolean functions $f\colon\pmset{n}\rightarrow \pmset{}$ naturally arise
in many areas of theoretical computer science and mathematics such as
learning theory, complexity theory, quantum computing, inapproximability,
graph theory, extremal combinatorics, etc. Fourier analysis over
the Boolean cube $\pmset{n}$ is a powerful technique that has been used
often to analyze problems in these areas. For a survey on the subject,
see~\cite{ODonnell-book,wolf:fouriersurvey}. One of the most important
and longstanding open problems in this field is the
\emph{Fourier Entropy-Influence} (FEI) conjecture,
first formulated by Ehud Friedgut and Gil Kalai in 1996~\cite{FK96}.
The FEI conjecture seeks to relate the following
two fundamental properties of a Boolean function $f\colon$
the \emph{Fourier entropy} of $f$ and
the \emph{total influence} of $f$, which we define now.

For a Boolean function $f\colon\pmset{n}\rightarrow \pmset{}$,
Parseval's identity relates the \emph{Fourier coefficients}
$\left\{\fc{f}{S}\right\}_S$ and the values $\left\{f(x)\right\}_x$ by
\begin{align*}
  \sum_{S\subseteq [n]} \fcsq{f}{S} = \E_x\left[f(x)^2\right] = 1,
\end{align*}
where the expectation is taken uniformly over the Boolean cube $\pmset{n}$.
An immediate implication of this equality is that
the squared Fourier coefficients
$\left\{\fcsq{f}{S}:S\subseteq [n]\right\}$ can be viewed as
a \emph{probability distribution} over subsets $S\subseteq [n]$,
which we often refer to as the \emph{Fourier distribution}.
The \emph{Fourier entropy of $f$} (denoted $\ent(\hat{f}^2)$)
is then defined as the Shannon entropy of the Fourier distribution, i.e.,
\begin{align*}
  \ent(\hat{f}^2) := \sum_{S\subseteq [n]} \fcsq{f}{S} \log\frac{1}{\fcsq{f}{S}}.
\end{align*}
The \emph{total influence of $f$} (denoted $\Inf(f)$)
measures the \emph{expected size} of a subset $S\subseteq [n]$,
where the expectation is taken according to the Fourier distribution, i.e.,
\begin{align*}
  \Inf(f)=\sum_{S\subseteq [n]} |S|~ \fcsq{f}{S}.
\end{align*}
Combinatorially $\Inf(f)$ is the same as the average sensitivity
$\as(f)$ of $f$. In particular, for $i \in [n]$, define $\Inf_i(f)$
to be the probability that on a uniformly random input
flipping the $i$-th bit changes the function value.
Then, $\Inf(f)$ is defined to be $\sum_{i=1}^n \Inf_i(f)$.

Intuitively, the Fourier entropy measures how ``spread out''
the Fourier distribution is over the $2^n$ subsets of $[n]$ and
the total influence measures the concentration of
the Fourier distribution on the ``high'' level coefficients.
Informally, the FEI conjecture states that Boolean functions
whose Fourier distribution is well ``spread out''
(i.e., functions with large Fourier entropy)
must have significant Fourier weight on the high-degree monomials
(i.e., their total influence is large).
Formally, the FEI conjecture can be stated as follows$\colon$
\begin{conjecture}[FEI Conjecture]
  \label{con:FEIC}
  There exists a universal constant $C>0$ such that
  for every Boolean function  $f\colon\pmset{n} \to \pmset{}$,   
  \begin{align}
    \label{eq:FEI-inequality}
    \ent(\hat{f}^2) \leq C\cdot \Inf(f).
  \end{align}
\end{conjecture}
The original motivation of Friedgut and Kalai for the FEI conjecture
came from studying threshold phenomena of monotone graph properties
in random graphs~\cite{FK96}. For example, resolving the FEI conjecture
would imply that every threshold interval of a monotone graph property
on $n$ vertices is of length at most $c(\log n)^{-2}$
(for some universal constant $c>0$). 
Bourgain and Kalai~\cite{BK97} proved a bound of
$c_\varepsilon(\log n)^{-2+\varepsilon}$  for every~$\varepsilon >0$.
A very recent work of Kelman et al.~\cite{KKLMS19} improved this bound
to $O((\log\log n)^2/(\log n)^2)$ by establishing nearly optimal
bounds on Fourier Min-entropy. In general, it implies sharper bounds
on the threshold interval of transitively symmetric functions.  

Besides this application, the FEI conjecture is known to imply
the famous Kahn-Kalai-Linial theorem~\cite{KKL88}
(otherwise referred to as the KKL theorem).
The KKL theorem was one of the first major applications of
Fourier analysis to understanding properties of Boolean functions
and has since found many application in various areas of
theoretical computer science.  
\begin{theorem}[KKL theorem]
  For every $f\colon\pmset{n}\rightarrow \pmset{}$,
  there exists an $i\in [n]$ such that
  \(\Inf_i(f)\geq \Var(f)\cdot \Omega\Big(\frac{\log n}{n} \Big)\),
  where $\Var(f) := 1 - \E_x[f(x)]^2 = 1 - \fcsq{f}{\emptyset}$.
\end{theorem} 
We discuss the implication of the FEI conjecture to the KKL theorem in more
detail in Section~\ref{sec:minentropy}.
Another motivation to study
the FEI conjecture is that a positive answer to this conjecture would resolve
the notoriously hard conjecture of Mansour~\cite{Mansour95} from 1995.
\begin{conjecture}[Mansour's conjecture]
  \label{con:Mansour-intro}
  Suppose $f\colon\pmset{n}\rightarrow \pmset{}$ is computed
  by a $t$-term $\dnf$.\footnote{A $t$-term $\dnf$ is a disjunction of at most $t$ conjunctions of variables and their negations.}
  Then for every $\varepsilon>0$,
  there exists a family $\mathcal T$ of subsets of $[n]$
  such that $|\mathcal T|\leq t^{O(1/\varepsilon)}$
  (i.e., size of $\mathcal{T}$ is polynomial in $t$)
  and \(\sum_{T\in \mathcal T} \fcsq{f}{T}\geq 1-\varepsilon\). 
\end{conjecture}
A positive answer to Mansour's conjecture, along with the query algorithm
of Gopalan et al.~\cite{GKK08}, would resolve a long-standing open
question in computational learning theory of agnostically learning $\dnf$s
under the uniform distribution in polynomial time (up to any constant accuracy).
We discuss this in more detail in Section~\ref{sec:Mansour}.

More generally, the FEI conjecture implies that every Boolean function 
can be approximated (in $\ell_2$-norm) by
\emph{sparse} polynomials over $\pmset{}$.
In particular, for a Boolean function $f$ and $\varepsilon > 0$,
the FEI conjecture implies the existence of a polynomial $p$
with $2^{O(\Inf(f)/\varepsilon)}$ monomials such that
\(\E_x\left[(f(x)-p(x))^2\right] \leq \varepsilon\).
The current best known bound in this direction is
$2^{O(\Inf(f)\log(\Inf(f))/\varepsilon)}$~\cite{KKLMS19}. 

Given the inherent difficulty in answering the FEI conjecture for arbitrary
Boolean functions, there have been many recent works studying the conjecture
for specific classes of Boolean functions. We give a brief overview of these
results in the next section. Alongside the pursuit of resolving
the FEI conjecture, O'Donnell and others~\cite{OWZ11,ODonnell-book}
have asked if a weaker question than the FEI conjecture,
the \emph{Fourier Min-entropy-Influence}
(FMEI) conjecture can be resolved. The FMEI conjecture asks if
the entropy-influence inequality in Eq.~\eqref{eq:FEI-inequality} holds  
when the entropy of the Fourier distribution is replaced by
the \emph{min-entropy} of the Fourier distribution
(denoted $\ent_{\infty}(\hat{f}^2)$).
The min-entropy of $\left\{\fcsq{f}{S}\right\}_S$ is defined as 
\begin{align*}
  \ent_{\infty}(\hat{f}^2) := \min_{\substack{S\subseteq [n]:\\ \fc{f}{S}\neq 0}}
  \Big\{ \log \frac{1}{\fcsq{f}{S}} \Big\}
\end{align*}
and thus it is easily seen that
\(\ent_{\infty}(\hat{f}^2)\leq \ent(\hat{f}^2)\).
In fact, $\ent_{\infty}(\hat{f}^2)$ could be much smaller compared
to $\ent(\hat{f}^2)$. For instance, consider the function
$f(x) := x_1 \vee  \mathsf{IP}(x_1,\ldots ,x_n)$; 
then $\ent_{\infty}(\hat{f}^2) = O(1)$ whereas
$\ent(\hat{f}^2) = \Omega(n)$.
($\mathsf{IP}$ is the inner-product-mod-2 function.)
So the FMEI conjecture could be strictly weaker than
the FEI conjecture, making it a natural candidate to resolve first.
\begin{conjecture}[FMEI Conjecture]
  \label{con:FMEIC}
  There exists a universal constant $C>0$ such that for every Boolean function
  \(f\colon\pmset{n} \rightarrow \pmset{}\), we have
  \begin{align*}\ent_{\infty}(\hat{f}^2)\leq  C \cdot \Inf(f).\end{align*}
\end{conjecture}
Another way to formulate the FMEI conjecture is,
suppose $f\colon\pmset{n}\to \pmset{}$,
then does there exist a Fourier coefficient $\fc{f}{S}$ 
such that \(|\fc{f}{S}|\geq 2^{-O(\Inf(f))}\)?
By the \emph{granularity} of Fourier coefficients
it is well-known that \emph{every} Fourier coefficient of
a Boolean function $f$ is an integral multiple of $2^{-\deg(f)}$, 
see~\cite[Exercise 1.11]{ODonnell-book} for a proof of this.
(Here the $\deg(f)$ refers to the degree of the unique
multilinear polynomial that represents $f$.)
The FMEI conjecture asks if we can prove a lower bound of
$2^{-O(\Inf(f))}$ on \emph{any one} Fourier coefficient,
and even this remains open. Proving the FMEI conjecture seems
to require proving interesting structural properties of Boolean functions.
It was observed by~\cite{OWZ11} that 
the FMEI conjecture suffices to imply the KKL theorem for balanced functions 
(see also Section~\ref{sec:minentropy}).
In fact, it suffices\footnote{With a minor strengthening: the large Fourier coefficient is required to be associated with a Fourier character of size $O(\Inf(f))$.} for most of the implications that follows
from the FEI conjecture (see \cite{KKLMS19}).

Understanding the min-entropy of a Fourier distribution is important
in its own right too. It was observed by Akavia et al.~\cite{ABGKR14} that
for a circuit class $\mathcal{C}$, tighter relations between min-entropy
of $f \in \mathcal{C}$ and $f_A$ defined as $f_A(x) := f(Ax)$,
for an arbitrary linear transformation $A$,
could enable us to translate lower bounds against the class
$\mathcal{C}$ to the class $\mathcal{C}\circ\mathsf{MOD_2}$. 
In particular, they conjectured that min-entropy of $f_A$ is only
polynomially larger than $f$ when
\(f \in \mathsf{AC}^0[\mathsf{poly}(n),O(1)]\).
(\textsf{AC}$^0[s,d]$ is the class of unbounded fan-in circuits
of size at most $s$ and depth at most $d$.)
It is well-known that when $f \in \mathsf{AC}^0[s,d]$,
$\ent_\infty(\hat{f}^2)$ is at most
\(O((\log s)^{d-1}\cdot \log\log s)\)~\cite{LMN93,Bop97,tal:fourierconcentration}.
Depending on the tightness of the relationship between
$\ent_\infty(\hat{f}^2)$ and $\ent_\infty(\widehat{f_A}^2)$,
one could obtain near-optimal lower bound on the size of
$\mathsf{AC^0}[s,d]\mathsf{\circ {MOD}_2}$ circuits computing
\textsf{IP} (inner-product-mod-2).
This problem has garnered a lot of attention in recent times
for a variety of reasons~\cite{SV10,SV12,ABGKR14,CS16,CGJWX18}.
The current best known lower bound for \textsf{IP} against
$\mathsf{AC^0}[s,d]\mathsf{\circ {MOD}_2}$ is quadratic when $d = 4$,
and only super-linear for all $d = O(1)$~\cite{CGJWX18}.  

In the remaining part of this introduction,
we first give an overview of prior work on the FEI conjecture
in Section~\ref{sec:previousworks} and then describe our contributions
in the following sections. 

\subsection{Prior work}
\label{sec:previousworks}
After Friedgut and Kalai~\cite{FK96} posed the FEI conjecture in 1996,
there was not much work done towards resolving it,
until the work of Klivans et al.~\cite{KLW10} in 2010.
They showed that the FEI conjecture holds true for
random $\dnf$ formulas. Since then, there have been many significant steps
taken in the direction of resolving the FEI conjecture.  We review some recent
works here, referring the interested reader to the blog post of
Kalai~\cite{Kalai-blog} for additional discussions on the FEI conjecture.

The FEI conjecture is known to be true when we replace
the universal constant $C$  with $\log n$ in Eq.~\eqref{eq:FEI-inequality}.
In fact we know \(\ent (\hat{f}^2)\leq O(\Inf(f)\cdot \log n)\)
for \emph{real-valued} functions $f\colon\pmset{n}\rightarrow \reals$
(see \cite{OWZ11,KMS12} for a proof and \cite{CKLS16}
for an improvement of this statement).\footnote{For Boolean functions, the $\log n$-factor was improved by \cite{gopalan:fouriertails} to $\log (\s(f))$ (where $\s(f)$ is the sensitivity of the Boolean function $f$).}
If we strictly require $C$ to be a universal constant,
then the FEI conjecture is known to be false for real-valued functions.
Instead, for real-valued functions an analogous statement called the
\emph{logarithmic Sobolev Inequality}~\cite{gross:logsobo} is known to be true.
The logarithmic Sobolev inequality states that
for every $f\colon\pmset{n}\rightarrow \reals$,
we have $\Ent(f^2)\leq 2 \cdot \Inf(f)$,
where $\Ent(f)$ is defined as $\Ent(f)=\E[f \ln (f)]-\E[f]\ln (\E[f])$,
where the expectation is taken over uniform $x\in \pmset{n}$. 

Restricting to Boolean functions,
the FEI conjecture is known to be true for the ``standard'' functions
that arise often in analysis, such as AND, OR, Majority, Parity,
Bent functions and Tribes.
There have been many works on proving the FEI conjecture
for specific classes of Boolean functions.
O'Donnell et al.~\cite{OWZ11} showed that
the FEI conjecture holds for symmetric Boolean functions and
read-once decision trees.
Keller et al.~\cite{KMS12} studied a generalization of
the FEI conjecture when the Fourier coefficients are
defined on biased product measures on the Boolean cube.
Then, Chakraborty et al.~\cite{CKLS16} and
O'Donnell and Tan~\cite{OT13}, independently and simultaneously, 
proved the FEI conjecture for
read-once formulas. In fact, O'Donnell and Tan proved an interesting
composition theorem for the FEI conjecture
(we omit the definition of composition theorem here,
see~\cite{OT13} for more).
For general Boolean functions, Chakraborty et al.~\cite{CKLS16} gave
several upper bounds on the Fourier entropy in terms of
combinatorial quantities larger than the total influence, e.g.,
average decision tree depth,~etc., and
sometimes even quantities that could be much smaller than influence,
namely, average parity-decision tree depth.  

Later Wan et al.~\cite{www14} used
Shannon's source coding theorem~\cite{Shannon48} (which characterizes entropy)
to establish the FEI conjecture for read-$k$ decision trees for constant $k$.
Using their novel interpretation of the FEI conjecture
they also reproved O'Donnell-Tan's composition theorem in an elegant way. 
Recently, Shalev~\cite{shalev2018:minentropy} improved the constant in the
FEI inequality for read-$k$ decision trees, and further verified the conjecture
when either the influence is \emph{too low}, or the entropy is \emph{too high}.
The FEI conjecture  is also verified for random Boolean functions by
Das et al.~\cite{das:randomfei} and for random linear threshold functions
(LTFs) by Chakraborty et al.~\cite{chakrabortyetal:feirandomltf}.

There has also been some work in giving lower bounds on the constant $C$ in
the FEI conjecture. Hod~\cite{hod:constantfei} gave a lower bound of
$C>6.45$
(the lower bound holds even when considering the class of monotone functions),
improving upon the lower bound of O'Donnell and Tan~\cite{OT13}.

However, there has not been much work on the FMEI conjecture. It was observed
in~\cite{OWZ11,chakrabortyetal:feirandomltf} that the KKL theorem implies
the FMEI conjecture for monotone functions and linear threshold functions.
Finally, the FMEI conjecture for ``regular'' read-$k$ $\dnf$s
was recently established by Shalev~\cite{shalev2018:minentropy}. 

\subsection{Our contributions}
\label{sec:contributions}
Our contributions in this paper are threefold, which we summarize
in the following sections.

\subsection{Better upper bounds for the FEI conjecture}
\label{sec:feiimprovement-intro}
Our first and main contribution of this paper is
to give a better upper bound on the Fourier entropy
$\ent(\hat{f}^2)$ in terms of $\aUC(f)$,
the \emph{average unambiguous certificate complexity} of $f$. 
Informally, the unambiguous certificate complexity $\UC(f)$ of
$f$ is similar to the standard certificate complexity measure,
except that the collection of certificates is now required to be
\emph{unambiguous}, i.e., every input should be consistent with
a \emph{unique} certificate. In other words, an unambiguous certificate is
a monochromatic subcube partition of the Boolean cube. By the average
unambiguous certificate complexity, $\aUC(f)$,
we mean the expected number of bits set
by an unambiguous certificate on a uniformly random input.

There have been many recent works on query complexity,
giving upper and lower bounds on $\UC(f)$ in terms of
other combinatorial measures such as decision-tree complexity,
sensitivity,  quantum query complexity, etc.,
see~\cite{goos:unambigious,ambainis:unambiguous,shalev:unambigious} for more.
It follows from definitions that $\UC(f)$ lower bounds
decision tree complexity. However, it is known that $\UC(f)$
can be quadratically smaller than
decision tree complexity~\cite{ambainis:unambiguous}.
Our main contribution here is an improved upper bound of
\emph{average unambiguous certificate complexity} $\aUC(f)$ on
$\ent(\hat{f}^2)$. This improves upon the previously known bound of
average decision tree depth on $\ent(\hat{f}^2)$~\cite{CKLS16}.
\begin{theorem} 
  \label{thm:subcube-partition-entropy-intro}
  Let $f\colon \pmset{n}\rightarrow \pmset{}$ be a Boolean function. Then,
  \begin{align*}
    \ent(\hat{f}^2) \leq 2\cdot \aUC(f).
  \end{align*}
\end{theorem}
A new and crucial ingredient employed in
the proof of the theorem is an analog of
the law of large numbers in information theory,
usually referred to as the
\emph{Asymptotic Equipartition Property (AEP)} theorem
(Theorem~\ref{thm:aep}).
Employing information-theoretic techniques for the FEI conjecture seems
very natural given that the conjecture seeks to bound the entropy of
a distribution. Indeed, Keller et al.~\cite[Section~3.1]{KMS12} envisioned 
a proof of the FEI conjecture itself using large deviation estimates and the
tensor structure (explained below) in a stronger way, 
and Wan et al.~\cite{www14} used
Shannon's source coding theorem~\cite{Shannon48}
to verify the conjecture for bounded-read decision trees.

In order to prove Theorem~\ref{thm:subcube-partition-entropy-intro},
we study the \emph{tensorized} version of $f$,
\(f^M \colon \pmset{Mn}\to \pmset{}\), which is defined as follows, 
\begin{align*}
  f^M(x^1,\ldots,x^M) 
  := f(x^1_{1},\ldots,x^1_{n})\cdot f(x^2_{1},\ldots,x^2_{n})\cdots f(x^M_{1},\ldots ,x^M_{n}).
\end{align*}
Similarly we define a \emph{tensorized} version $\mathscr{C}^M$ of
an unambiguous certificate $\mathscr{C}$ of $f$,\footnote{Recall an unambiguous certificate is a collection of certificates that partitions the Boolean cube $\pmset{n}$.} i.e., a direct product of
$M$ independent copies of $\mathscr{C}$. 
It is not hard to see that $\mathscr{C}^M$ is also
an unambiguous certificate of $f^M$. To understand
the properties of $\mathscr{C}^M$ we study $\mathscr{C}$
in a probabilistic manner. We observe that $\mathscr{C}$ naturally induces 
a distribution $\mathbf{C}$ on its certificates when
the underlying inputs $x\in \pmset{n}$ are distributed uniformly. 
Using the asymptotic equipartition property with respect to $\mathbf{C}$,
we infer that for every $\delta > 0$, there exists $M_0>0$ such that
for all $M \geq M_0$, there are at most $2^{M(\aUC(f,\mathscr{C}) + \delta)}$  
certificates in $\mathscr{C}^M$ that together cover at least $1-\delta$ 
fraction of the inputs in $\pmset{Mn}$.
Furthermore, each of these certificates fixes at most
$M(\aUC(f,\mathscr{C}) + \delta)$ bits. Hence, a particular certificate
can contribute to at most $2^{M(\aUC(f,\mathscr{C}) + \delta)}$ Fourier coefficients
of $f^M$. Combining both these bounds, all these certificates
can overall contribute to at most $2^{2M(\aUC(f,\mathscr{C}) + \delta)}$
Fourier coefficients of $f^M$.
Let's denote this set of Fourier coefficients by $\mathcal{B}$.
We then argue that the Fourier coefficients of $f^M$
that are \emph{not} in $\mathcal{B}$ have Fourier weight at most $\delta$.
This now allows us to bound the Fourier entropy of $f^M$ as follows, 
\begin{align*}
  \ent(\widehat{f^M}^2) \leq \log |\mathcal{B}| + \delta nM + \mathsf{H}(\delta),
\end{align*}
where $\mathsf{H}(\delta)$ is the binary entropy function.
Since \(\ent(\widehat{f^M}^2) = M\cdot \ent(\hat{f}^2)\), we have
\begin{align*}
  \ent(\hat{f}^2) \leq 2(\aUC(f,\mathscr{C}) + \delta) + \delta n + \frac{\mathsf{H}(\delta)}{M}.
\end{align*}
By the AEP theorem, note that $\delta \to 0$ as $M \to \infty$.
Thus, taking the limit as $M \to \infty$ we obtain our theorem. 

Looking finely into how certificates contribute to Fourier coefficients
in the proof above, we further strengthen
Theorem~\ref{thm:subcube-partition-entropy-intro}
by showing that we can replace $\aUC(f)$ by the
\emph{average unambiguous parity-certificate complexity}
$\aUC^{\oplus}(f)$ of $f$. Here $\aUC^{\oplus}(f)$ is defined similar
to $\aUC(f)$ except that instead of being defined in terms of
monochromatic subcube partitions of $f$, we now partition
the Boolean cube with monochromatic \emph{affine subspaces}.
(Observe that subcubes are also affine subspaces.)
This strengthening also improves upon the previously known bound of
average parity-decision tree depth on $\ent(\hat{f}^2)$~\cite{CKLS16}.
It is easily seen that $\aUC^{\oplus}(f)$ lower bounds
the average parity-decision tree depth. 
\begin{theorem}
  \label{thm:subspace-partition-entropy-intro}
  Let $f\colon \pmset{n}\to \pmset{}$ be any Boolean function. Then,
  \begin{align*}\ent(\hat{f}^2) \leq 2\cdot \aUC^{\oplus}(f).\end{align*} 
\end{theorem}
The proof outline remains the same as in
Theorem~\ref{thm:subcube-partition-entropy-intro}.
However, a particular certificate in $\mathscr{C}^M$ no longer fixes
just variables. Instead these parity certificates now fix parities
over variables, and so potentially could involve all variables.
Hence we cannot directly argue that all the certificates contribute
to at most $2^{M(\aUC^{\oplus}(f,\mathscr{C})+\delta)}$ Fourier coefficients of $f^M$.
Nevertheless, by the AEP theorem we still obtain that
a typical parity-certificate fixes at most
$M(\aUC^{\oplus}(f,\mathscr{C})+\delta)$ parities.
Looking closely at the Fourier coefficients that
a parity-certificate can contribute to, we now argue
that such coefficients must lie in the linear span of
the parities fixed by the parity-certificate. 
Therefore, a typical parity-certificate can overall
contribute to at most $2^{M(\aUC^{\oplus}(f,\mathscr{C})+\delta)}$
Fourier coefficients of $f^M$.
The rest of the proof now follows analogously.

\subsubsection{Further extension to unambiguous $\dnf$s}
Consider an unambiguous certificate
$\mathscr{C} = \{C_1,\ldots ,C_t\}$  of $f$. It covers both $1$
and $-1$ inputs of $f$. Suppose $\{C_1,\ldots ,C_{t_1}\}$ for some $t_1 < t$
is a partition of $f^{-1}(-1)$ and $\{C_{t_1+1}, \ldots ,C_t\}$
is a partition of $f^{-1}(1)$. To represent $f$, it suffices to consider
$\bigvee_{i=1}^{t_1}C_i$. This is a $\dnf$ representation of $f$
with the additional property that  $\{C_1,\ldots ,C_{t_1}\}$
forms a partition of $f^{-1}(-1)$. We call such a representation
an \emph{unambiguous} $\dnf$. In general, a $\dnf$ representation
need not satisfy this additional property. 

Using the equivalence of total influence and average sensitivity,
one can easily observe that 
\begin{align*}
  \Inf(f) \leq 2 \cdot \min \left\{\sum_{i=1}^{t_1} \mathsf{co\text{-}dim}(C_i) \cdot 2^{-\mathsf{co\text{-}dim}(C_i)}, \sum_{i=t_1+1}^{t} \mathsf{co\text{-}dim}(C_i) \cdot 2^{-\mathsf{co\text{-}dim}(C_i)}\right\} \leq \aUC(f,\mathscr{C}),
\end{align*} 
where $\mathsf{co\text{-}dim}(\cdot)$ denotes
the co-dimension of an affine space. Note that the quantity
\(\sum_{i=1}^{t_1} \mathsf{co\text{-}dim}(C_i)\cdot 2^{-\mathsf{co\text{-}dim}(C_i)}\),
in a certain sense, is  ``average unambiguous $1$-certificate complexity'' and,
similarly,
\(\sum_{i=t_1+1}^{t} \mathsf{co\text{-}dim}(C_i) \cdot 2^{-\mathsf{co\text{-}dim}(C_i)}\)
captures ``average unambiguous $0$-certificate complexity''.

Building on our ideas from the main theorem in the previous section
and using a \emph{stronger} version (Theorem~\ref{thm:strongaep})
of the AEP theorem we nearly (see Remark~\ref{rem:one-sided-bound}) establish
the aforementioned improved bound of the smaller quantity between 
``average unambiguous $1$-certificate complexity''
and ``average unambiguous $0$-certificate complexity'' on the Fourier entropy. 
Formally, we prove the following.
\begin{theorem}
  \label{thm:one-sided-subspace-partition-intro}
  Let $f\colon\pmset{n} \to \pmset{}$ be a Boolean function and
  $\mathscr{C} = \{C_1, \ldots ,C_t\}$ be a monochromatic affine
  subspace partition of $\pmset{n}$ with respect to $f$ such that
  $\{C_1,\ldots , C_{t_1}\}$ for some $t_1 < t$ is an affine subspace partition
  of $f^{-1}(-1)$ and $\{C_{t_1 +1} , \ldots , C_t\}$ is an affine subspace
  partition of $f^{-1}(1)$. Further, $p := \Pr_x[f(x) =1]$. Then,
  \begin{align*}
    \ent(\hat{f}^2) \leq
    \begin{dcases*}
      2\left(\sum_{i=1}^{t_1} \mathsf{co\text{-}dim}(C_i) \cdot 2^{-\mathsf{co\text{-}dim}(C_i)} + p \cdot \max_{i \in \{1,\ldots ,t_1\}}\mathsf{co\text{-}dim}(C_i)\right), \\ 
      2\left(\sum_{i=t_1+1}^{t} \mathsf{co\text{-}dim}(C_i) \cdot 2^{-\mathsf{co\text{-}dim}(C_i)} + (1-p) \cdot \max_{i \in \{t_1+1,\ldots ,t\}}\mathsf{co\text{-}dim}(C_i)\right).
    \end{dcases*}
  \end{align*}
\end{theorem}
The outline for the proof of
Theorem~\ref{thm:one-sided-subspace-partition-intro}
remains the same as before, but it differs in implementation details. 
We sketch them now. Analogous to the proof of the main theorem
we consider a partition of inputs with respect to $f$ and
its tensorized version. Motivated by the $\dnf$ representation,
we study the following partition $\{C_1,\ldots , C_{t_1}, f^{-1}(1)\}$
which naturally inherits a distribution $\mathbf{C}$ given by
the uniform distribution on the underlying inputs.
Again we build a ``small'' set $\mathcal{B}$ of Fourier coefficients of $f^M$
based on the Fourier expansions of strongly typical sequences.
However, unlike before, the probability of observing
a strongly typical sequence doesn't capture the number of coefficients
it could contribute to $\mathcal{B}$. 
Here, we use stronger properties guaranteed by the strong AEP.
In particular, it guarantees that the \emph{empirical} distribution
of a typical sequence is close to the distribution of $\mathbf{C}$.
In contrast, the (weak) AEP only guarantees that
the \emph{empirical} entropy of a typical sequence is close to
the entropy of $\mathbf{C}$. Using the stronger property
we can now lower bound the magnitude of any non-zero Fourier coefficient
in the Fourier expansion of the indicator function of
a strongly typical sequence. We then use Parseval's Identity
(Fact~\ref{fact:plancherel}) to deduce an upper bound on its Fourier sparsity,
which in turn is used to bound the size of $\mathcal{B}$.
We also need to argue that coefficients \emph{not} in $\mathcal{B}$
have negligible Fourier weight, which can be done as before.
Using the two properties, we can now complete the proof. 

We also note that a similar bound for the general $\dnf$ representation,
i.e., when $\{C_1,\ldots ,C_{t_1}\}$ is an arbitrary DNF representation
of $f$ where the $C_i$s need not be disjoint,
suffices to establish Mansour's conjecture
(Conjecture~\ref{con:Mansour-intro}). In fact, following the analogy,
Theorem~\ref{thm:one-sided-subspace-partition-intro} implies
a bound of ``average $1$-certificate complexity'' in the general case.
In this direction, in Section~\ref{sec:Mansour} we observe that
a weaker bound of $1$-certificate complexity, i.e.,
showing $\ent(\widehat{f}^2)\leq O(\mathsf{C}_1(f))$,
would already suffice to answer Mansour's conjecture positively.
For details, see Section~\ref{sec:Mansour}.

\subsection{New upper bounds for the FMEI conjecture}
\label{sec:fmei-intro}
Given the hardness of obtaining better upper bounds on
the Fourier entropy of a Boolean function, we make progress on
a weaker conjecture, the FMEI conjecture. The FMEI conjecture is
much less studied than the FEI conjecture. In fact, we are aware of
only one very recent paper~\cite{shalev2018:minentropy}
which studies the FMEI conjecture for a particular class of functions.
Our second contribution is to give upper bounds on the min-entropy of
general Boolean functions in terms of the minimum parity-certificate complexity
(denoted $\certificate_{\min}^{\oplus}(f)$) and the approximate spectral norm of
Boolean functions (denoted $\|\widehat{f}\|_{1,\varepsilon}$).
The minimum parity-certificate complexity $\certificate_{\min}^{\oplus}(f)$
is also referred to as the parity kill number by
O'Donnell et al.~\cite{O'donnell:paritykill} and
is closely related to the communication complexity of
XOR functions~\cite{zhangshi:quantumxor,montanaro:xor,tsang:logrank}.
The approximate spectral norm $\|\widehat{f}\|_{1,\varepsilon}$ is related to
the \emph{quantum} communication complexity of XOR
functions~\cite{leeshraibman:lowerbounds,zhang:quantumxor}.
In particular, it characterizes the bounded-error
quantum communication complexity of XOR functions
with constant $\mathbb{F}_2$-degree~\cite{zhang:quantumxor}.
(By $\mathbb{F}_2$-degree, we mean the degree of a function
when viewed as a polynomial over $\mathbb{F}_2$.)
\begin{theorem}\label{thm:minentropyupperb-intro}
  Let $f\colon\pmset{n}\rightarrow \pmset{}$ be a Boolean function. Then,  
  \begin{enumerate}[label=(\roman*)]
  \item\label{min-vs-approx-norm-intro} For every $\varepsilon \geq 0$,~~
    \(\ent_{\infty}(\hat{f}^2)\leq 2\cdot\log \left(\|\widehat{f}\|_{1,\varepsilon}/(1-\varepsilon)\right)\).
  \item\label{min-vs-cert-intro} \(\ent_{\infty}(\hat{f}^2)\leq  2\cdot \certificate_{\min}^{\oplus}(f)\). 
  \item\label{min-vs-rpdt-intro} \(\ent_{\infty}(\hat{f}^2)\leq  2(1+\log_2 3)\cdot R_2^\oplus(f)\).\footnote{$R_2^\oplus(f)$ is the randomized parity-decision tree complexity of $f$ (we define this formally in Section~\ref{sec:prelims}).} 
  \end{enumerate}
\end{theorem}
The proof of Theorem~\ref{thm:minentropyupperb-intro}~\ref{min-vs-approx-norm-intro} expresses the quantity $\|\widehat{f}\|_{1,\varepsilon}$
as a (minimization) linear program.
We consider the dual linear program and exhibit a feasible solution that
achieves an optimum of \((1-\varepsilon)/\max_S|\fc{f}{S}|\).
This proves the desired inequality. In order to prove
part~\ref{min-vs-cert-intro} and \ref{min-vs-rpdt-intro} of the theorem,
the idea is to consider a ``simple'' function $g$
that has ``good'' correlation with $f$, and then upper bound
the correlation between $f$ and $g$
using Plancherel's theorem (Fact~\ref{fact:plancherel}) and
the fact that $g$ has a ``simple'' Fourier structure.
For part~\ref{min-vs-cert-intro},
$g$ is chosen to be the indicator function of
an (affine) subspace where $f$ is constant, whereas for
part~\ref{min-vs-rpdt-intro} the randomized parity decision tree
computing~$f$ itself plays the role of $g$.\footnote{We remark here that there exists simpler proof of part~\ref{min-vs-approx-norm-intro}, along the lines of parts~\ref{min-vs-cert-intro} and \ref{min-vs-rpdt-intro}. However, we believe that the linear-programming formulation of $\ent_{\infty}(\hat{f}^2) $ might help obtain better bounds, such as fractional block sensitivity.}

As a corollary (Corollary~\ref{coro:renyientropyupperb}) of this theorem
we also obtain upper bounds on the R\'{e}nyi Fourier entropy
$\ent_{1+\delta}(\hat{f}^2)$ of order $1+\delta$ for all $\delta>0$.
We recall that $\ent_{1+\delta}(\hat{f}^2) \geq \ent_{\infty}(\hat{f}^2)$
for every $\delta\geq 0$ and as $\delta \to \infty$,~
$\ent_{1+\delta}(\hat{f}^2)$ converges to $\ent_{\infty}(\hat{f}^2)$.
Also $\ent_1(\hat{f}^2)$, obtained by taking the limit as $\delta \to 0$,
is the standard Shannon entropy.

We believe that these improved bounds on min-entropy of
the Fourier distribution give a better understanding of Fourier coefficients
of Boolean functions, and could be of
independent interest. As a somewhat non-trivial application of
Theorem~\ref{thm:minentropyupperb-intro}
(in particular, part~\ref{min-vs-cert-intro}) we verify the FMEI conjecture
for read-$k$ $\dnf$s, for constant~$k$.
(A read-$k$ $\dnf$ is a formula where each variable appears
in at most $k$ terms.)
\begin{theorem}
  \label{thm:meic-dnf-intro}
  For every Boolean function $f \colon \pmset{n}\to \pmset{}$
  that can be expressed as a read-$k$ $\dnf$, we have  
  \begin{align*}
    \ent_{\infty}(\hat{f}^2) ~\leq~ O(\log k) \cdot \Inf(f).
  \end{align*}
\end{theorem}
This theorem improves upon a recent (and independent) result of
Shalev~\cite{shalev2018:minentropy} that establishes the FMEI conjecture
for ``regular" read-$k$ $\dnf$s (where regular means each term in
the $\dnf$ has more or less the same number of variables,
see~\cite{shalev2018:minentropy} for a precise definition).  
In order to prove Theorem~\ref{thm:meic-dnf-intro}, we essentially
show that $\Inf(f)$ is at least as large as $\certificate_{\min}(f)$,
for read-$k$ $\dnf$s (Lemma~\ref{lem:dnf-inf-lb}),
where $\certificate_{\min}(f)$ is the minimum certificate complexity of $f$.
Now the proof of Theorem~\ref{thm:meic-dnf-intro} follows in conjunction
with Theorem~\ref{thm:minentropyupperb-intro}~\ref{min-vs-cert-intro}.  

\paragraph{\textbf{New research work since the completion of this work}:}
A very recent work of Kelman et al.~\cite{KKLMS19} shows that
the FMEI conjecture holds up to a logarithmic factor in $\Inf(f)$.
In general they show that almost all the Fourier weight lies on
Fourier coefficients of size at most $O(\widetilde{\Inf(f)})$
and weight at least
\(2^{-O\left(\widetilde{\Inf(f)}\log\left(1+ \widetilde{\Inf(f)}\right)\right)}\),
where \(\widetilde{\Inf(f)} := \frac{\Inf(f)}{\Var(f)}\).
We remark that though these results
are very close to the FMEI conjecture, they are incomparable to our bounds.
In particular one could construct examples where
that $\certificate_{\min}^{\oplus}(f)$ is much smaller 
than \(\widetilde{\Inf(f)}\log\left(1+\widetilde{\Inf(f)}\right)\) 
(e.g., the inner-product-mod-2 function, AND, etc.)
and vice versa (e.g., Majority).

However we remark that there is a similarity in the high-level
proof idea of their work and ours. In order to capture
a Fourier coefficient of large weight consider
the inner product between the function $f$ and another function $g$.
In our case we choose $g$ to be a subfunction (of $f$) that has \emph{large}
inner product with $f$, whereas Kelman et al.~\cite{KKLMS19}
essentially choose $g$ to be the truncation of the Fourier spectrum
of $f$ to the Fourier coefficients of size at most  $O(\widetilde{\Inf(f)})$
and weight \emph{at most}
\(2^{-\Omega\left(\widetilde{\Inf(f)}\log\left(1+\widetilde{\Inf(f)}\right)\right)}\).
They then present a non-trivial analysis to show that
the inner product between $f$ and $g$ is \emph{very} small.
This is sufficient to reach the conclusion since 
the total Fourier weight on characters of size
at most $O(\widetilde{\Inf(f)})$, by Markov's inequality, 
is a large constant (at least $1 - \Var(f)/100$).

\subsection{Implications of the FEI conjecture and connections to the Bohnenblust-Hille inequality}
\label{sec:bh-intro}
Our final contribution is to better understand the structure of polynomials
that $\varepsilon$-approximate Boolean functions on the Boolean cube.
To be more specific, for simplicity we fix $\varepsilon=1/3$ and we
consider polynomials $p$ such that $|p(x)-f(x)| \leq 1/3$ for all
$x \in \pmset{n}$, where $f$ is a Boolean function. Such polynomials have
proved to be powerful and found diverse applications in
theoretical computer science. The single most important measure associated
with such polynomials is its \emph{degree}. The \emph{least} degree of
a polynomial that $1/3$-approximates $f$ is referred to as the
\emph{approximate degree} of $f$. Tight bounds on approximate degree
have both algorithmic and complexity-theoretic~implications,
see for instance Sherstov's recent paper~\cite{Sherstov18}
and references therein.

In this work we ask, suppose the FEI conjecture were true, what can be said
about approximating polynomials?
For instance, are these approximating polynomials $p$
sparse in their Fourier domain, i.e., is the number of monomials in $p$,
$|\{S \colon \fc{p}{S} \neq 0\}|$, small? Do approximating polynomials have
small spectral norm (i.e., small $\sum_S |\fc{p}{S}|$)?
In order to understand these questions better, we restrict
ourselves to a class of polynomials called \emph{flat} polynomials
over $\pmset{n}$, i.e., polynomials whose non-zero Fourier coefficients
have the same magnitude.

We first observe that if a flat polynomial $p$ $1/3$-approximates a
Boolean function $f$, then the entropy of the Fourier distribution of $f$
must be ``large''. In particular, we show that
$\ent(\hat{f}^2)$ must be at least as large as
the logarithm of the Fourier sparsity of $p$ (Claim~\ref{cl:entflat}).
It then follows that assuming the FEI conjecture,
a flat polynomial of degree $d$ and sparsity $2^{\omega(d)}$
cannot $1/3$-approximate a Boolean function (Lemma~\ref{lem:flatpoly}).
However, it is not clear to us how to obtain the same conclusion
\emph{unconditionally} (i.e., without assuming that the FEI conjecture
is true) and, so we pose the following conjecture. 
\begin{conjecture}
  \label{conj:introsparsity}
  No flat polynomial of degree $d$ and sparsity $2^{\omega(d)}$
  can $1/3$-approximate a Boolean function.
\end{conjecture}
\begin{remark}
  We remark that there exist degree-$d$ flat \emph{Boolean} functions
  of sparsity $2^d$. One simple example on $4$ bits is
  the function $x_1(x_2+x_3)/2+ x_4(x_2-x_3)/2$.
  By taking a $(d/2)$-fold product of this Boolean function
  on disjoint variables, we obtain our remark.
\end{remark}
Since we could not solve the problem as posed above, we make progress
in understanding this conjecture by further restricting ourselves to the
class of \emph{block-multilinear} polynomials. An $n$-variate polynomial
is said to be \emph{block-multilinear} if the input variables can be
\emph{partitioned} into disjoint blocks 
such that every monomial in the polynomial has \emph{at most}
one variable from each block.
Such polynomials have been well-studied in functional analysis
since the work of Bohnenblust and Hille~\cite{bohnenblust:unbalancing},
but more recently have found applications in
quantum computing~\cite{aaronsonambainis:forrelation,montanaro:hypercontractivity}, classical and quantum XOR games~\cite{briet:XOR}, and 
polynomial decoupling~\cite{Odonnell:decoupling}.
In the functional analysis literature homogeneous block-multilinear
polynomials are known as \emph{multilinear forms}. 
In an ingenious work~\cite{bohnenblust:unbalancing},
Bohnenblust and Hille showed that for every
degree-$d$ multilinear form $p\colon(\reals^n)^d\to \reals$, we  have
\begin{align}
  \label{eq:BHinequality-intro}
  \Big(\sum_{i_1,\ldots,i_d=1}^n |\widehat{p}_{i_1,\ldots,i_d}|^{\frac{2d}{d+1}}\Big)^{\frac{d+1}{2d}}\leq C_d \cdot  \max_{x^1,\ldots,x^d\in [-1,1]^n} |p(x^1,\ldots,x^d)|,
\end{align}
where $C_d$ is a constant that depends on $d$.
In~\cite{bohnenblust:unbalancing}, they showed that
it suffices to pick $C_d$ to be exponential in $d$
to satisfy the equation above.
For $d=2$, Eq.~\eqref{eq:BHinequality-intro} generalizes
Littlewood's famous $4/3$-inequality~\cite{littlewood:inequality}.
Eq.~\eqref{eq:BHinequality-intro} is commonly referred to as
the Bohnenblust-Hille (BH) inequality and is known to have
deep applications in various fields of analysis such as
operator theory, complex analysis, etc.
There has been a long line of work on improving the constant $C_d$
in the BH inequality (to mention a few~\cite{defant1:unbalancing,defant:unbalancing,albuquerque:unbalancing,BPSS14,pellergrino:sharpconstant}).
The best known upper bound on $C_d$ (we are aware of)
is polynomial in $d$. It is also conjectured that
it suffices to let $C_d$ be a \emph{universal} constant
(independent of $d$) in order to satisfy Eq.~\eqref{eq:BHinequality-intro}.

In our context, using the best known bound on $C_d$ in the
BH-inequality implies that a flat homogeneous block-multilinear polynomial
of degree $d$ and sparsity $2^{\omega(d\log d)}$ cannot
$1/3$-approximate a Boolean function. However, from the discussion
before Conjecture~\ref{conj:introsparsity},
we know that the FEI conjecture implies the following theorem.
\begin{theorem}
  \label{thm:unbalancing2-intro}
  If $p$ is a flat block-multilinear polynomial of degree $d$ and sparsity $2^{\omega(d)}$, then $p$ cannot $1/3$-approximate a Boolean function.
\end{theorem}
Moreover, the above theorem, with an added restriction of homogeneity on $p$,   
is also implied when the BH-constant $C_d$ is
assumed to be a universal constant.
Our main contribution here is to establish the above theorem
\emph{unconditionally}, i.e., neither assuming $C_d$ is
a universal constant nor assuming the FEI conjecture.
In order to show the theorem,
we show an inherent weakness of block-multilinear polynomials
in approximating  Boolean functions. More formally, we show the following.
\begin{lemma}
  Let $p$ be a block-multilinear polynomial of degree-$d$
  that $1/3$-approximates a Boolean function $f$. Then, \(\deg(f) \leq d\).
\end{lemma}

\paragraph{Organization} 
In the remainder of the paper, we prove and elaborate on each of
these results in more detail. We start with preliminaries in
Section~\ref{sec:prelims}. In Section~\ref{sec:feiimprovement},
we give improved upper bounds on Fourier entropy
(Theorems~\ref{thm:subcube-partition-entropy-intro},~\ref{thm:subspace-partition-entropy-intro} and~\ref{thm:one-sided-subspace-partition-intro}),
verify the FEI conjecture for functions
with bounded average unambiguous certificate complexity,
and elaborate on the connections between Fourier entropy
vs.~certificate complexity and Mansour's conjecture.
In Section~\ref{sec:minentropy}, we establish new upper bounds on
Fourier min-entropy (Theorem~\ref{thm:minentropyupperb-intro}),
and verify the FMEI conjecture for read-$k$ $\dnf$s
(Theorem~\ref{thm:meic-dnf-intro}).
In Section~\ref{sec:unbalancinglights}, we pose a conjecture
(Conjecture~\ref{conj:introsparsity}) which is a consequence of
the FEI conjecture and make partial progress
(Theorem~\ref{thm:unbalancing2-intro}) towards its resolution.
We further discuss an intriguing connection between our conjecture
and the constants in the Bohnenblust-Hille inequality.
Finally we conclude with some open problems in Section~\ref{sec:conclusion}.

\section{Preliminaries}
\label{sec:prelims}
\paragraph{Notation}
We denote the set $\{1,2,\ldots ,n\}$ by $[n]$.
A \emph{partial assignment} of $[n]$ is a map
\(\tau \colon [n] \to \{-1,1,\ast\}\).
Define \(|\tau|=|\tau^{-1}(1) \cup \tau^{-1}(-1)|\).
A subcube of the Boolean cube $\pmset{n}$ is a set of $x\in \pmset{n}$
that agrees with some partial assignment $\tau$, i.e.,
\(\{x\in \pmset{n} :  x_i = \tau(i) \text{ for every }i \text{ with }
\tau(i)\neq \ast\}\).

\paragraph{Fourier Analysis}
We recall some definitions and basic facts from analysis of Boolean functions,
referring to~\cite{ODonnell-book, wolf:fouriersurvey} for more.
Consider the space of all functions from $\pmset{n}$ to $\reals$ equipped
with the inner product defined as
\begin{align*}
  \langle f,g\rangle \coloneqq \E_x[f(x)g(x)] =
  \frac{1}{2^n}\sum_{x\in \pmset{n}} f(x)g(x).
\end{align*}
For $S\subseteq [n]$, the character function
\(\chi_S\colon\pmset{n}\rightarrow \pmset{}\) is defined as
\(\chi_S(x) \coloneqq \prod_{i\in S}x_i\).
Then the set of character functions $\left\{\chi_S\right\}_{S\subseteq [n]}$
forms an orthonormal basis for the space of
all real-valued functions on $\pmset{n}$.
Hence, every real-valued function
$f \colon\pmset{n}\rightarrow \reals$ has a unique Fourier expansion
\begin{align*}
  f(x) = \sum_{S \subseteq [n]} \fc{f}{S} \chi_S(x).
\end{align*} 
The \emph{degree} of $f$, denoted $\deg(f)$, is defined as
\(\max\left\{|S|:\fc{f}{S}\neq 0\right\}\).
The \emph{spectral norm} of $f$ is defined to be
$\sum_S|\fc{f}{S}|$. The \emph{Fourier weight} of a function $f$
on a set $\mathcal{S}$ of coefficients is defined as
$\sum_{S\in \mathcal{S}} \widehat{f}(S)^2$.
The \emph{approximate spectral norm} of a Boolean function $f$ is defined as
\begin{align*}
  \|\widehat{f}\|_{1,\varepsilon}=\min\Big\{\sum_S|\widehat{p}(S)| : |p(x)-f(x)|\leq \varepsilon \text{ for every } x\in \pmset{n}\Big\}.
\end{align*}
We note a well-known fact that follows from
the orthonormality of the character functions.
\begin{fact}[Plancherel's Theorem]
  \label{fact:plancherel}
  For any $f,g \colon \pmset{n} \to \reals$,  
  \begin{align*}\E_x[f(x)g(x)] = \sum_{S\subseteq [n]}\fc{f}{S}\fc{g}{S}.\end{align*}
\end{fact}
In particular, if $f\colon\pmset{n}\rightarrow \pmset{}$ is
Boolean-valued and $g = f$,
we have Parseval's Identity \(\sum_{S}\fcsq{f}{S} = \E\left[f(x)^2\right]\),
which in turn equals $1$.
Hence $\sum_{S}\fcsq{f}{S}=1$ and we can view $\left\{\fcsq{f}{S}\right\}_S$
as a probability distribution, which allows us to discuss
the Fourier entropy and min-entropy of the distribution
$\left\{\fcsq{f}{S}\right\}_S$, defined as 
\begin{definition}
  \label{def:shannon-entropy}
  For a Boolean function $f\colon\pmset{n}\rightarrow \pmset{}$,
  its Fourier entropy (denoted $\ent(\hat{f}^2)$) and
  min-entropy (denoted $\ent_\infty(\hat{f}^2)$) are 
  \begin{align*}
    \ent(\hat{f}^2) \coloneqq \sum_{S \subseteq [n]}\fcsq{f}{S} \log \frac{1}{\fcsq{f}{S}}, \quad \mbox{and} \quad \ent_\infty(\hat{f}^2) \coloneqq \min_{\substack{S \subseteq [n]:\\\widehat{f}(S)\neq 0}} \Big\{\log \frac{1}{\fcsq{f}{S}}\Big\}.
  \end{align*}
\end{definition}
Similarly, we can also define the R\'{e}nyi Fourier entropy.
\begin{definition}[R\'{e}nyi Fourier entropy]
  \label{def:renyi-entropy}
  For $f\colon\pmset{n}\rightarrow \pmset{}$, $\alpha \geq 0$
  and $\alpha \neq 1$, the R\'{e}nyi Fourier entropy of $f$
  of order $\alpha$ is defined as
  \begin{align*}
    \ent_{\alpha}(\hat{f}^2) \coloneqq \frac{1}{1-\alpha}\log\left(\sum_{S\subseteq [n]}|\fc{f}{S}|^{2\alpha}\right).
  \end{align*}
\end{definition}
It is known that in the limit as $\alpha \to 1$, 
$\ent_{\alpha}(\hat{f}^2)$ is the (Shannon) Fourier entropy
$\ent(\hat{f}^2)$ (see~\cite[Chapter~17, Section~8]{CT91}) and
when $\alpha \to \infty$,
observe that $\ent_{\alpha}(\hat{f}^2)$ converges to $\ent_{\infty}(\hat{f}^2)$.
It is easily seen  that
\begin{align*}
  \ent_{\infty}(\hat{f}^2) \leq \ent_1(\hat{f}^2) \leq
  \ent_{\frac{1}{2}}(\hat{f}^2) \leq \ent_0(\hat{f}^2).
\end{align*}
For $f\colon\pmset{n}\rightarrow \pmset{}$,
the \emph{influence} of a coordinate $i\in [n]$, denoted $\Inf_i(f)$,
is defined as
\begin{align*}
  \Inf_i(f)=\Pr_{x \in \sbool^n} \left[ f(x) \neq f(x^{(i)})\right]
  = \E_x\Big[\Big(\frac{f(x)-f(x^{(i)})}{2}\Big)^2\Big],
\end{align*}
where the probability and expectation is taken according to
the uniform distribution on $\pmset{n}$ and
$x^{(i)}$ is $x$ with the $i$-th bit flipped.
The \emph{total influence} of $f$, denoted $\Inf(f)$, is 
\begin{align*}
  \Inf(f)=\sum_{i\in [n]} \Inf_i(f).
\end{align*}
In terms of the Fourier coefficients of $f$,
it can be shown, e.g.,~\cite{KKL88}, that
\(\Inf_i(f)=\sum_{S\ni i} \widehat{f}(S)^2\), and therefore 
\begin{align*}
  \Inf(f)=\sum_{S\subseteq [n]} |S|\widehat{f}(S)^2.
\end{align*}
The \emph{variance} of a real-valued function $f$ is given by
\(\Var(f):=\sum_{S\neq \emptyset} \widehat{f}(S)^2 = 1 -\fcsq{f}{\emptyset}\).
It easily follows that $\Var(f)\leq \Inf(f)$. 
We now introduce some basic complexity measures of Boolean functions
which we use often, referring to~\cite{buhrman&wolf:dectreesurvey} for more. 
\paragraph{Sensitivity}
For $x \in \pmset{n}$, the \emph{sensitivity} of $f$ at $x$,
denoted $\s_f(x)$, is defined to be the number of neighbors
$y$ of $x$ in the Boolean hypercube
(i.e., $y$ is obtained by flipping \emph{exactly} one bit of $x$)
such that $f(y)\neq f(x)$. The sensitivity $\s(f)$ of $f$ is
$\max_{x}\{\s_f(x)\}$. The \emph{average sensitivity} $\as(f)$ of $f$
is defined to be $\E_x[\s_f(x)]$. By the linearity of expectation observe that 
\begin{align*}
  \E_{x}\left[\s_f(x)\right]=\sum_{i=1}^n \Pr_x\left[f(x)\neq f(x^{(i)})\right]
  =\sum_{i=1}^n\Inf_i(f)=\Inf(f),
\end{align*}
so the average sensitivity of $f$ equals the total influence of $f$.
As a result, the FEI conjecture asks if
\(\ent(\hat{f}^2)\leq C\cdot  \as(f)\) for every Boolean function $f$.
\paragraph{Certificate complexity}
For $x \in \pmset{n}$,
the \emph{certificate complexity} of $f$ at $x$,
denoted $\certificate(f,x)$,
is the minimum number of bits in $x$ that needs to be fixed
to ensure that the value of $f$ is constant.
The certificate complexity $\certificate(f)$ of $f$ is
$\max_x\{\certificate(f,x)\}$.
The minimum certificate complexity of $f$ is
\(\certificate_{\min}(f)=\min_x\{\certificate(f,x)\}\).
The $0$-certificate complexity $\certificate^{0}(f)$ of $f$ is
$\max_{x : f(x) = 1}\{\certificate(f,x)\}$.
Similarly, the $1$-certificate complexity
$\certificate^{1}(f)$ of $f$ is
$\max_{x : f(x) = -1}\{\certificate^{1}(f,x)\}$.
Observe that for every $x\in \pmset{n}$,
$\s(f,x)\leq \certificate(f,x)$.
This gives $\s(f)\leq \certificate(f)$ and
$\as(f)\leq \ac(f)$ where $\ac(f)$ denotes
the \emph{average} certificate complexity of $f$.
As before, the average is taken with respect to
the uniform distribution on $\pmset{n}$. 
\paragraph{Parity-certificate complexity}
Analogously, we  define the \emph{parity-certificate complexity}
$\certificate^{\oplus}(f,x)$ of $f$ at $x$ as the minimum number of
parities on the input variables one has to fix in order to fix 
the value of $f$ at $x$, i.e., 
\begin{align*}
  \certificate^{\oplus}(f,x) \coloneqq \min\{\mathsf{co\text{-}dim}(H) \mid H \text{ is an affine subspace on which }f\text{ is constant and }x \in H \},
\end{align*} 
where $\mathsf{co\text{-}dim}(H)$ is
the \emph{co-dimension} of the affine subspace $H$.
It is easily seen that \(\certificate^{\oplus}(f,x) \leq \certificate(f,x)\).
We also define
\(\certificate^{\oplus}(f) \coloneqq \max_x\{\certificate^{\oplus}(f,x)\}\),
and
\(\certificate_{\min}^{\oplus}(f) \coloneqq \min_x\{\certificate^{\oplus}(f,x)\}\).
\paragraph{Unambiguous certificate complexity}
We now define the \emph{unambiguous certificate complexity} of $f$.
Let $\tau\colon [n]\rightarrow\{-1,1,\ast \}$ be a partial assignment.
We refer to
\begin{align*}
  S_\tau=\left\{x\in \pmset{n}: x_i=\tau(i) \text{ for every } i\in [n]\setminus \tau^{-1}(\ast)\right\}
\end{align*}
as the subcube generated by $\tau$.
We call $C \subseteq \pmset{n}$ a \emph{subcube} of $\pmset{n}$
if there exists a partial assignment $\tau$ such that
$C=S_\tau$ and the {co-dimension} of $C$ is
the number of bits fixed by $\tau$, i.e.,
\(\mathsf{co\text{-}dim}(C)=|\{i\in [n]:\tau(i)\neq \ast\}|\). 
A set of subcubes $\mathscr{C}= \{C_1,\ldots ,C_m\}$ \emph{partitions}
$\pmset{n}$ if the subcubes are disjoint and they cover $\pmset{n}$,
i.e., $C_i\cap C_j=\emptyset$ for $i\neq j$ and $\cup_{i} C_i=\pmset{n}$.  

An \emph{unambiguous certificate} $\mathcal{U}=\{C_1,\ldots,C_m\}$,
also referred to as a \emph{subcube partition},
is a set of subcubes partitioning $\pmset{n}$.
We say $\mathcal{U}$ computes a Boolean function
$f\colon\pmset{n}\rightarrow \pmset{}$ if $f$ is constant on each $C_i$
(i.e., $f(x)$ is the same for all $x\in C_i$).
For an unambiguous certificate $\mathcal{U}$,
the \emph{unambiguous certificate complexity} $\UC(\mathcal{U},x)$
on input $x$ equals $\mathsf{co\text{-}dim}(C_i)$
for the $C_i$ satisfying $x\in C_i$.
Define the \emph{average unambiguous certificate complexity} of $f$
with respect to $\mathcal{U}$ as
\begin{align*}
  \aUC(f,\mathcal{U}) \coloneqq \E_x[\UC(\mathcal{U},x)].
\end{align*}
Then, the \emph{average unambiguous certificate complexity} of $f$
is defined as  
\begin{align*}
  \aUC(f)\coloneqq \min_{\mathcal{U}} \aUC(f,\mathcal{U}),
\end{align*}
where the minimization is over all unambiguous certificates for $f$.
Finally, the \emph{unambiguous certificate complexity of $f$} is
\begin{align*}
  \UC(f) \coloneqq \min_{\mathcal{U}} \max_x \UC(\mathcal{U},x).
\end{align*}
Note that since unambiguous certificates are more restricted
than general certificates, we have $\certificate(f)\leq \UC(f)$.

An unambiguous $\oplus$-certificate
$\mathcal{U} = \{C_1,\ldots ,C_m\}$ for $f$ is defined to be
a collection of monochromatic \emph{affine subspaces}
that together partition the space $\pmset{n}$.
It is easily seen that a subcube is also an affine subspace.
Analogously, for an unambiguous $\oplus$-certificate $\mathcal{U}$,
on an input $x$, $\UC^\oplus(\mathcal{U},x) := \mathsf{co\text{-}dim}(C_i)$
for the $C_i$ satisfying $x\in C_i$, and
\(\aUC^\oplus(f,\mathcal{U}) := \E_x[\UC^\oplus(\mathcal{U},x)] \).
Similarly, we define $\aUC^\oplus(f)$ and $\UC^\oplus(f)$. 
\paragraph{$\dnf$s}
A $\dnf$ (disjunctive normal form) is
a disjunction ($\bor$) of conjunctions ($\band$s) of variables and
their negations. An \emph{unambiguous $\dnf$} is a $\dnf$ that
satisfies the additional property that$\colon$
on every $(-1)$-input, exactly one of the conjunctions outputs $-1$.
\paragraph{Approximate degree}
The $\varepsilon$-\emph{approximate degree} of
$f\colon\pmset{n}\rightarrow \reals$, denoted ${\deg}_{\varepsilon}(f)$,
is defined to be the minimum degree among all multilinear real polynomials
$p$ such that $|f(x)-p(x)| \leq \varepsilon$ for all $x \in \pmset{n}$.
Usually $\varepsilon$ is chosen to be $1/3$,
but it can be chosen to be any constant in $(0,1)$,
without significantly changing the model. 
\paragraph{Deterministic decision trees}
A deterministic decision tree for
$f\colon\pmset{n}\rightarrow \pmset{}$ is a rooted binary tree
where each node is labelled by $i\in [n]$ and the leaves are labelled
with an output bit $\pmset{}$. On input $x\in \pmset{n}$,
the tree proceeds at a node with label $i$ by evaluating the bit $x_i$
and continuing with the subtree corresponding to the value of $x_i$.
Once a leaf is reached, the tree outputs a bit.
We say that a deterministic decision tree computes $f$
if for all $x \in \pmset{n}$ its output equals $f(x)$. 

A parity-decision tree for $f$ is similar to a deterministic decision tree,
except that each node in the tree is labelled by a subset $S\subseteq [n]$.
On input $x\in \pmset{n}$, the tree proceeds at a node with label $S$
by evaluating the parity of the bits $x_i$ for $i\in S$ and
continues with the subtree corresponding to the value of
$\oplus_{i\in S}x_i$. Note that if the subsets at each node
have size $|S|=1$, then we get the standard deterministic decision tree model.
\paragraph{Randomized decision trees}
A randomized decision tree for $f$ is a probability distribution $\mu$
over deterministic decision trees for $f$.
On input $x$, a decision tree is chosen according to $\mu$,
which is then evaluated on $x$.
The complexity of the randomized tree is
the largest depth among all deterministic trees with non-zero probability
of being sampled according to $\mu$.
One can similarly define a randomized parity-decision tree
as a probability distribution over
deterministic parity-decision trees for $f$.

We say that a randomized decision tree computes $f$ with bounded-error 
if for all $x \in \pmset{n}$ its output equals $f(x)$ with probability
at least $2/3$. $R_2(f)$ (respectively, $R_2^\oplus(f)$)
denotes the complexity of the optimal randomized (respectively, parity)
decision tree that computes $f$ with bounded-error,
i.e., errs with probability at most $1/3$.
\paragraph{Information Theory}
We need some preliminaries from Information theory.
We start with the following consequence of the law of large numbers,
called the \emph{Asymptotic Equipartition Property (AEP)} or
\emph{the Shannon-McMillan-Breiman theorem}. See Chapter~3 in
the book \cite{CT91} for more details. 
\begin{theorem}[Asymptotic Equipartition Property (AEP) Theorem]
  \label{thm:aep}
  Let $\mathbf{X}$ be a random variable drawn from a distribution $P$ and
  suppose $\mathbf{X_1, X_2, \ldots ,X_M}$ are independently and
  identically distributed copies of\, $\mathbf{X}$, then 
  \begin{align*}
  -\frac{1}{M} \log P(\mathbf{X}_1,\mathbf{X}_2,\ldots ,\mathbf{X}_M)
  \longrightarrow \ent(\mathbf{X})
  \end{align*}
  in probability as $M \to \infty$. 
\end{theorem}
We now recall the notion of typical sequences.  
\begin{definition}
  \label{def:typical-set}
  Fix $\varepsilon\geq 0$. The \emph{typical set}
  $T^{(M)}_\varepsilon(\mathbf{X})$ with respect to a distribution $P$
  is defined to be the set of sequences
  $(x_1,x_2,\ldots ,x_M) \in \mathbf{X_1 \times X_2 \times \cdots \times X_M}$
  such that 
  \begin{align*}
    2^{-M(\mathbb{H}(\mathbf{X})+\varepsilon)} \leq P(x_1,x_2,\ldots ,x_M) \leq 2^{-M(\mathbb{H}(\mathbf{X}) - \varepsilon)}.
  \end{align*}
\end{definition}
The following properties of the typical set follows from the AEP. 
\begin{theorem}[{\cite[Theorem~3.1.2]{CT91}}]
  \label{thm:typical-set}
  Let $\varepsilon\geq 0$ and $T^{(M)}_\varepsilon(\mathbf{X})$
  be a typical set with respect to $P$, then
  \begin{enumerate}[label=(\roman*)]
  \item\label{ts-i} \( |T^{(M)}_\varepsilon(\mathbf{X})| \leq 2^{M(\mathbb{H}(\mathbf{X})+\varepsilon)} \). 
  \item\label{ts-ii} Suppose $x_1,\ldots,x_M$ are drawn i.i.d.~according
    to $\mathbf{X}$, then   
    \begin{align*}
      \Pr \left[(x_1,\ldots,x_M)\in T^{(M)}_\varepsilon(\mathbf{X})\right] \geq 1 - \varepsilon,
    \end{align*}
    for $M$ sufficiently large. 
  \item \(|T^{(M)}_\varepsilon(\mathbf{X})| \geq (1-\varepsilon)2^{M(\mathbb{H}(\mathbf{X})-\varepsilon)}\), for $M$ sufficiently large. 
  \end{enumerate}
\end{theorem}
We also require the following stronger version of typical sequences
and asymptotic equipartition property.    
\begin{definition}[{\cite[Chapter~11, Section~2]{CT91}}]
  \label{def:strongtypical-set}
  Let $\mathbf{X}$ be a random variable drawn according to
  a distribution $P$. Fix $\varepsilon > 0$.
  The \emph{strongly typical set} $T_{\varepsilon}^{\ast (M)}(\mathbf{X})$
  is defined to be the set of sequences
  \(\rho = (x_1,x_2,\ldots ,x_M) \in \mathbf{X_1 \times X_2 \times \cdots \times X_M}\) such that for every $x$, $N(x;\rho)=0$ if $P(x) = 0$, and otherwise
  \begin{align*}
    \left | \frac{N(x;\rho)}{M}- P(x) \right | \leq \frac{\varepsilon}{|\mathbf{X}|},
  \end{align*}
  where $N(x;\rho)$ is defined as the number of occurrences of $x$ in $\rho$.
\end{definition}
The strongly typical set shares similar properties with
its (weak) typical counterpart which we state now.
See \cite[Chapter~11, Section~2]{CT91} for a proof of this theorem. 
\begin{theorem}[Strong AEP Theorem]
  \label{thm:strongaep}
  Following the notation in Definition~\ref{def:strongtypical-set},
  let $T^{\ast (M)}_\varepsilon(\mathbf{X})$ be a strongly typical set.
  Then, there exists $\delta > 0$ such that $\delta \to 0$ as
  $\varepsilon \to 0$, and the following hold$\colon$   
  \begin{enumerate}[label=(\roman*)]
  \item\label{sts-i} Suppose $x_1,\ldots,x_M$ are drawn i.i.d.~according to
    $\mathbf{X}$, then 
    \begin{align*}
      \Pr\left[(x_1,\ldots,x_M)\in T^{\ast (M)}_\varepsilon(\mathbf{X})\right] \geq 1 - \varepsilon,
    \end{align*}
    for $M$ sufficiently large.
  \item\label{sts-ii} If $(x_1,\ldots,x_M) \in T^{\ast (M)}_\varepsilon(\mathbf{X})$, then 
    \begin{align*}
      2^{-M(\mathbb{H}(\mathbf{X})+\delta)} \leq P(x_1,\ldots,x_M) \leq 2^{-M(\mathbb{H}(\mathbf{X})-\delta)}.
    \end{align*}
  \item\label{sts-iii} For $M$ sufficiently large, 
    \begin{align*}
      (1-\varepsilon)2^{M(\mathbb{H}(\mathbf{X})-\delta)} \leq |T^{\ast (M)}_\varepsilon(\mathbf{X})| \leq 2^{M(\mathbb{H}(\mathbf{X})+\delta)}.
    \end{align*}
  \end{enumerate} 
\end{theorem}

\section{Better bounds on Fourier entropy}
\label{sec:feiimprovement}
In this section we prove the main result of this paper$\colon$ 
a new improved upper bound on the Fourier entropy of Boolean functions.
It is well-known that $\Inf(f)$ lower bounds
many combinatorial measures associated with Boolean functions
such as decision tree depth, certificate complexity, sensitivity, etc.
Given the difficulty in resolving the FEI conjecture,
it is natural to wonder if  the Fourier entropy can be
upper bounded by these larger measures.
Indeed, Chakraborty et al.~\cite{CKLS16} established many bounds on
the Fourier entropy, including \emph{average parity-decision tree complexity}.
We improve on their bounds by showing an upper bound of
\emph{average unambiguous parity-certificate complexity}.
In order to keep the presentation clear, we first prove a weaker upper bound of
\emph{average unambiguous certificate complexity} $\aUC(f)$ on
$\mathbb{H}(\hat{f}^2)$. We then explain how to generalize the proof
to establish the stronger upper bound of
\emph{average unambiguous parity-certificate complexity} $\aUC^\oplus(f)$.

\subsection{Average unambiguous certificate complexity}
We recall an unambiguous certificate $\mathscr{C} = \{C_1,\ldots ,C_t\}$
for $f$ is a collection of monochromatic subcubes (with respect to $f$)
that together partition the hypercube $\pmset{n}$.
The \emph{average unambiguous certificate complexity}
of $f$ with respect to $\mathscr{C}$, denoted by $\aUC(f,\mathscr{C})$,
equals $\E_{x \in \pmset{n}}\left[\UC(\mathscr{C},x)\right]$. Further,
\(\aUC(f) = \min_{\mathscr{C}}\left\{\aUC(f,\mathscr{C})\right\}\).
We now prove the main theorem. 
\begin{theorem}[Restatement of Theorem~\ref{thm:subcube-partition-entropy-intro}]
  \label{thm:partition-entropy}
  For every Boolean function $f \colon \pmset{n}\rightarrow \pmset{}$, we have
  \begin{align*}
    \ent(\hat{f}^2) \leq 2\cdot \aUC(f).
  \end{align*}
\end{theorem}
Before we prove this inequality, we first give a sketch of the proof. 
Given $f \colon \pmset{n} \rightarrow \pmset{}$ and an unambiguous certificate 
$\mathscr{C}$ for $f$, our first step is to consider
the function $f^M\colon \pmset{Mn} \rightarrow \pmset{}$ defined as
the $M$-fold product of $f$. 
We then consider a random variable $\mathbf{C}$ (supported on $\mathscr{C}$)
and let $T^{(M)}_\delta(\mathbf{C})$ be a typical set associated with
$M$ i.i.d copies of $\mathbf{C}$. 
Based on the typical set, we define a $2^{2M\cdot (\aUC(f) + \delta)}$-sized
set $\mathcal{B}$ of Fourier coefficients of $f^M$ and show that
$\mathcal{B}$ has fairly large Fourier weight. 
We then consider the re-normalized entropy when restricted to
the Fourier coefficients in $\mathcal{B}$. 
By taking the limit ($M\rightarrow \infty$),
we obtain $\ent(\hat{f}^2)\leq O(\aUC(f))$.
We now fill in the details.
\begin{proof}
  It suffices to show that for every unambiguous certificate $\mathscr{C}$
  for $f$, we have
  \begin{align*}
    \ent(\hat{f}^2) \leq 2\cdot\aUC(f,\mathscr{C}).
  \end{align*}
  Let \(\mathscr{C} := \{C_1,\ldots,C_t\}\) be an
  unambiguous certificate of $f$.
  For every $C_i$, let $\tau(C_i)$ be the partial assignment
  \(\tau(C_i)\colon [n]\rightarrow \{-1,1,\ast\}\),
  corresponding to the bits fixed by $C_i$.
  Let $\supp(\tau(C_i))\subseteq [n]$ be the set of indices
  where $\tau(C_i)\neq \ast$. As a shorthand, we denote
  $|\supp(\tau(C_i))|$ by $|\tau(C_i)|$. 
  Consider the Boolean function \(f^M\colon \pmset{Mn}\rightarrow \pmset{}\)
  given by the \emph{M-fold iterated product} of $f$
  with itself over distinct variables, i.e.,
  \begin{align*}
    f^M(x^1,\ldots,x^M) 
    = f(x^1_{1},\ldots,x^1_{n})\cdot f(x^2_{1},\ldots,x^2_{n})\cdots f(x^M_{1},\ldots ,x^M_{n}),
  \end{align*}
  where \(x^i\in \pmset{n}\) for every \(i\in [M]\).
  First, observe that \(\ent(\widehat{f^M}^2) = M \cdot \ent(\hat{f}^2)\).
  Similarly, we also have
  \(\aUC(f^M,\mathscr{C}^M) = M \cdot \aUC(f,\mathscr{C})\).

  We now bound the Fourier entropy of $f^M$ by showing that there is
  a ``small'' set of Fourier coefficients of $f^M$
  whose total Fourier weight is approximately $1$. 
	
  Let $\mathbf{C}$ be a subcube-valued random variable that equals $C_i$ with
  probability $2^{-|\tau(C_i)|}$.\footnote{Since $\{C_1,\ldots,C_t\}$ are disjoint subcubes partitioning $\pmset{n}$, we have that $\sum_{i=1}^t2^{-|\tau(C_i)|}=1$.}
  Further, let  \(\mathbf{C}_1,\mathbf{C}_2,\ldots , \mathbf{C}_M\)
  be i.i.d random copies of $\mathbf{C}$. For a choice of $\delta>0$,
  let $T^{(M)}_\delta(\mathbf{C})$ be the typical set
  with respect to the distribution \(\left\{2^{-|\tau(C_i)|}\right\}_{i\in [t]}\).
  We now define a set $\mathcal{B}$ of Fourier coefficients of $f^M$,
  which we argue below to have large Fourier weight. 
  \begin{align}
    \label{eq:defnB}
    \mathcal{B} :=\Big\{ (S_1, \ldots ,S_M) \subseteq [n]^M \mid S_i\subseteq \supp(\tau(C_i)) \subseteq [n] \text{ for } i\in [M] \text{ and } (C_1, \ldots , C_M)\in T^{(M)}_\delta(\mathbf{C})  \Big\}.
  \end{align}
  Using Theorem~\ref{thm:typical-set} about typical sequences,
  we are now ready to bound the size of $\mathcal{B}$ as follows. 
  \begin{claim}
    \label{claim:good-sets}
    \begin{align*}
      \left|\mathcal{B}\right| \leq 2^{2M(\aUC(f,\mathscr{C})+\delta )}.
    \end{align*}
  \end{claim}
  \begin{proof}[Proof of Claim]
    We first bound the size of $T^{(M)}_\delta(\mathbf{C})$ and
    then count contributions of a typical sequence $(C_1,\ldots ,C_M)$
    to $\mathcal{B}$.
    For the first bound, by the properties of the AEP
    Theorem~\ref{thm:typical-set}~\ref{ts-i},
    the total number $|T^{(M)}_\delta(\mathbf{C})|$ of typical sequences 
    is at most $2^{ M(\ent(\mathbf{C})+\delta )}$. For the second bound,
    observe that $(C_1,\ldots ,C_M)$ contributes a set
    $(S_1,\ldots ,S_M)$ to $\mathcal{B}$ if and only if
    $S_i \subseteq \supp(\tau(C_i))$ for all $i \in [M]$. Therefore,
    the maximum possible contribution of a typical sequence is bounded by  
    \begin{align*}
      2^{|\tau(C_1)|+\cdots +|\tau(C_M)|} =
      \left(\Pr[\mathbf{C}_1=C_1,\ldots,\mathbf{C}_M=C_M]\right)^{-1} \leq
      2^{ M(\ent(\mathbf{C})+\delta )},
    \end{align*}
    where the equality is because the random variable $\mathbf{C}$
    is sampled according to the distribution
    $\left\{2^{-|\tau(C_i)|}\right\}_{i\in[t]}$ and the inequality follows
    from Definition~\ref{def:typical-set} of typical sets.
  
    Combining both the upper bounds, we get
    \(\left| \mathcal{B}\right| \leq 2^{2M(\ent(\mathbf{C}) + \delta)}\).
    Finally, by the definition of entropy, we have
    \begin{align*}
      \ent(\mathbf{C})=\sum_{i=1}^t 2^{-|\tau(C_i)|}\cdot |\tau(C_i)|=
      \frac{1}{2^n}\sum_{x\in \pmset{n}} \UC(\mathscr{C},x) = \aUC(f,\mathscr{C}),
    \end{align*}
    where the second equality used that $\{C_1,\ldots,C_t\}$
    formed an unambiguous certificate for $f$.
  \end{proof}
  
  We now claim that $\mathcal{B}$ is a ``small'' set of
  Fourier coefficients of $f^M$ that has ``large'' Fourier weight.
  In order to quantitatively prove this, we show that
  the Fourier coefficients that are \emph{not} in $\mathcal{B}$
  have total Fourier weight at most $\delta$. In what follows,
  with some abuse of notation,
  we will denote a set $S \subseteq [n]$ and the corresponding subset
  $\{x_i \mid i \in S\}$ of variables by the same set $S$.  
  \begin{claim}
    \label{claim:bad-sets}
    \begin{align*}
      \sum_{(S_1,S_2,\ldots,S_M) \not\in \mathcal{B}} \fcsq{f^M}{S_1\cup\cdots \cup S_M} \leq \delta.
    \end{align*}
  \end{claim}
  \begin{proof}[Proof of Claim]
    We saw earlier that $\mathscr{C}^M$ is an unambiguous certificate of $f^M$.
    Let $\rho \in \mathscr{C}^M$ be a certificate of $f^M$,
    and $\mathds{1}_{\rho}(z)$ be the $\bool$-valued function
    that is $1$ if and only if $z$ is consistent with the certificate $\rho$.
    Further we denote the value $f^M$ takes on any input consistent
    with $\rho$ by $f^M(\rho)$. We can then  express $f^M$ on an input
    $z\in \pmset{Mn}$ as follows
  \begin{align}
    \label{eq:f-expansion}
    f^M(z)  = \sum_{\rho \in \mathscr{C}^M} f^M(\rho) \cdot \mathds{1}_\rho(z) 
    = \sum_{\rho \in T^{(M)}_\delta(\mathbf{C})} f^M(\rho) \cdot \mathds{1}_\rho(z) + \underbrace{\sum_{\rho \not\in T^{(M)}_\delta(\mathbf{C})} f^M(\rho) \cdot \mathds{1}_\rho(z)}_{:=g(z)}.
  \end{align}
  For $(S_1,\ldots,S_M)\subseteq [n]^M$,
  consider the expansion of the Fourier coefficient
  \begin{align*}
    \fc{f^M}{S_1\cup \cdots\cup S_M} & = \E_{z} \left[f^M(z) \chi_{S_1\cup\cdots\cup S_M}(z)\right]\\
    &=\E_{z} \left[\sum_{\rho \in T^{(M)}_\delta(\mathbf{C})} f^M(\rho) \cdot \mathds{1}_\rho(z)\cdot \chi_{S_1\cup\cdots\cup S_M}(z)  + \sum_{\rho \not\in T^{(M)}_\delta(\mathbf{C})} f^M(\rho) \cdot \mathds{1}_\rho(z) \cdot \chi_{S_1\cup\cdots\cup S_M}(z)\right] \\
    & = \sum_{\rho \in T^{(M)}_\delta(\mathbf{C})} f^M(\rho) \cdot \E_z\left[ \mathds{1}_\rho(z)\cdot \chi_{S_1\cup\cdots\cup S_M}(z)\right]  + \sum_{\rho \not\in T^{(M)}_\delta(\mathbf{C})} f^M(\rho) \cdot \E_z\left[\mathds{1}_\rho(z) \cdot \chi_{S_1\cup\cdots\cup S_M}(z)\right]. 
  \end{align*}
  Now observe that for a fixed certificate $\rho$, we have
  \(\E_z\left[ \mathds{1}_\rho(z)\cdot \chi_{S_1\cup\cdots\cup S_M}(z)\right] \neq 0\)
  if and only if $\rho$ fixes the variables in $S_1\cup\cdots\cup S_M$. 
  By definition of $\mathcal{B}$ it now follows that if
  \((S_1,\ldots ,S_M) \notin \mathcal{B}\), then
  \(\E_z\left[ \mathds{1}_\rho(z)\cdot \chi_{S_1\cup\cdots\cup S_M}(z)\right] = 0\) 
  for \(\rho \in T^{(M)}_\delta(\mathbf{C})\),
  and thus \(\fc{f^M}{S_1\cup\cdots \cup S_M}\) gets contribution
  only from $\rho$ that are \emph{not} typical, i.e.,
  \(\rho \not\in T^{(M)}_{\delta}(\mathbf{C})\). 
  
  In this direction, consider the function $g(z)$ defined in
  Eq.~\eqref{eq:f-expansion}, which is $\{-1,1,0\}$-valued.
  Using the argument above, we have that if
  \((S_1,\ldots, S_M) \not\in \mathcal{B}\), then
  \(\fc{f^M}{S_1\cup\cdots \cup S_M} = \fc{g}{S_1\cup\cdots \cup S_M}\).
  Then, clearly,   
  \begin{align}
    \label{eq:sumofFourcoeffnotintypical}
    \sum_{(S_1,S_2,\ldots,S_M) \not \in \mathcal{B}} 
    \fcsq{f^M}{S_1\cup\cdots \cup S_M} \leq \sum_T \fcsq{g}{T}.
  \end{align}
  Moreover by Parseval's Theorem (Fact~\ref{fact:plancherel}),
  \(\sum_T \fcsq{g}{T} = \E_z\left[g(z)^2\right]\). Therefore,
  \begin{align*}
    \sum_{(S_1,S_2,\ldots,S_M) \not \in \mathcal{B}} 
    \fcsq{f^M}{S_1\cup\cdots \cup S_M} \leq \E_z\left[g(z)^2\right] =\Pr_{z}\left[ z\notin T^{(M)}_\delta(\mathbf{C})\right] \leq \delta,   
  \end{align*}
  where the first inequality uses Eq.~\eqref{eq:sumofFourcoeffnotintypical}
  and Parseval, the second equality is because $g(z)^2\in \{0,1\}$ and the
  last inequality follows from Theorem~\ref{thm:typical-set}~\ref{ts-ii}.
  \end{proof}
   We are now ready to bound the Fourier entropy of $f$ and
  complete the proof. We need the following well-known trick
  to bound entropy when the underlying distribution has large weight
  on a small support. Let $\mathcal{S}$ be a set of Fourier characters
  such that $\sum_{S \not\in \mathcal{S}} \fcsq{g}{S} = \delta$.
  In order to upper bound the Fourier entropy,
  we first express $\ent(\hat{g}^2)$ as follows
  \begin{align*}
  \ent(\hat{g}^2)=\sum_{S\in \mathcal{S}} \fcsq{g}{S} \log \Big(\frac{1}{\fcsq{g}{S}}\Big)+\sum_{S\notin \mathcal{S}} \fcsq{g}{S} \log \Big(\frac{1}{\fcsq{g}{S}}\Big).
  \end{align*}
  We renormalize the first expression in the sum by $(1-\delta)$ and
  the second expression by $\delta$. By doing so, we get
  \begin{align}
    \label{eq:fourier-entropy}
    \ent(\hat{g}^2) &= (1-\delta)\ent\left(\frac{\fcsq{g}{S}}{1-\delta} \colon S \in \mathcal{S}\right) 
      + \delta \ent\left(\frac{\fcsq{g}{S}}{\delta} \colon S\not\in \mathcal{S}\right) - \sum_{S\in \mathcal S}\widehat{g}(S)^2 \log (1-\delta) - \sum_{S\notin \mathcal S}\widehat{g}(S)^2 \log (\delta), \nonumber \\
      &=(1-\delta)\ent\left(\frac{\fcsq{g}{S}}{1-\delta} \colon S \in \mathcal{S}\right) 
      + \delta \ent\left(\frac{\fcsq{g}{S}}{\delta} \colon S\not\in \mathcal{S}\right) + \mathsf{H}(\delta), 
  \end{align}
  where the equality used \(\sum_{S\in \mathcal S}\fcsq{g}{S} = 1 - \delta\)
  and we denote \(\mathsf{H}(p) :=p\log\frac{1}{p} +(1-p)\log\frac{1}{1-p}\). 

  We now use Eq.~\eqref{eq:fourier-entropy} when applied to
  the function $f^M$ and set  $\mathcal{S}=\mathcal{B}$. We then obtain
  \begin{align*}
    M \cdot \ent(\hat{f}^2) = \ent(\widehat{f^M}^2) & \leq \log |\mathcal{B}|+ \delta\cdot \log |\{S: S\notin \mathcal{B}\}|+ \mathsf{H}(\delta) \leq 2 M(\aUC(f,\mathscr{C})+\delta) + \delta nM + \mathsf{H}(\delta),
  \end{align*}
  where the second inequality used Claims~\ref{claim:good-sets}
  and~\ref{claim:bad-sets}. Dividing by $M$ on both sides, we get 
  \begin{align*}
    \ent(\hat{f}^2) \leq  2\cdot(\aUC(f,\mathscr{C}) + \delta)+ \delta n + \frac{\mathsf{H}(\delta)}{M}.  
  \end{align*}
  By the AEP theorem we know that $\delta \to 0$ as $M \to \infty$.
  Therefore, allowing $M \to \infty$ and taking the limit gives us the theorem.
\end{proof}

\subsubsection{Extension to $\oplus$-certificate complexity}
\label{sec:parity-certificate}
We now discuss a strengthening of Theorem~\ref{thm:partition-entropy}
where we improve the upper bound to 
\emph{average unambiguous parity-certificate complexity} $\aUC^{\oplus}(f)$. 
Recall that an unambiguous $\oplus$-certificate
$\mathscr{C} = \{C_1,\ldots ,C_t\}$ for $f$ is
a collection of monochromatic \emph{affine subspaces}
that together partition the space $\pmset{n}$. 
(Observe that a subcube is a special type of affine subspace.)
Analogously, the \emph{average unambiguous $\oplus$-certificate complexity}
of $f$ with respect to $\mathscr{C}$, denoted by $\aUC^{\oplus}(f,\mathscr{C})$,
equals \(\E_x[\UC^{\oplus}(\mathscr{C},x)]\)
and \(\aUC^{\oplus}(f) \coloneqq \min_{\mathscr{C}} \aUC^{\oplus}(f,\mathscr{C})\).
Let $A_{i}$ be the set of parities fixed by $C_{i}$ for $i \in [t]$. A parity
is defined over a subset of variables and thus, naturally, can be viewed
as a vector in $\bool^n$.

Like in the proof of Theorem~\ref{thm:partition-entropy},
we study the $M$-fold iterated product of $\mathscr{C}$.
In order to find a ``small'' set of coefficients where
the Fourier weight is concentrated, we define $\mathcal{B}$ differently.
The Fourier expansion of $f^M$, given by Eq.~\eqref{eq:f-expansion}, 
suggests the following definition.

For a set $S \subseteq [n]$, define $\mathds{1}_{S} \in \bool^n$
to be the indicator vector representing $S$ 
(i.e., $\mathds{1}_{S}(j)=1$ if and only if $j\in S$).  
Let $(S_1,\ldots ,S_M)$ be an $M$-tuple where each $S_i\subseteq [n]$.  
Then, we define $\mathcal{B}$ by letting
$(S_1,\ldots,S_M) \in \mathcal{B}$ if and only if there exists
a typical sequence \((C_{i_1},\ldots ,C_{i_M}) \in T^{(M)}_\delta(\mathbf{C})\) 
such that for all $j \in [M]$,
\(\mathds{1}_{S_j} \in \mathsf{span}\langle A_{i_j}\rangle\)
(where by $\mathsf{span}\langle A_{i_j} \rangle$,
we mean the linear $\mathbb{F}_2$-span of parities in $A_{i_j}$,
when viewed as vectors).
We recall that $A_{i_j}$ is the set of parities fixed by $C_{i_j}$. 
Observe that the earlier definition of $\mathcal{B}$ is now a special case of
this. 
With this definition of $\mathcal{B}$ the rest of the proof
follows similarly to establish the following generalization.
\begin{theorem}[Restatement of Theorem~\ref{thm:subspace-partition-entropy-intro}]
  \label{thm:subspace-partition-entropy}
  Let $f \colon\pmset{n}\rightarrow \pmset{}$ be any Boolean function. Then,
  \begin{align*}
    \ent(\hat{f}^2) \leq 2\cdot \aUC^{\oplus}(f).
  \end{align*}
\end{theorem}
We remark that as a corollary to the theorem it follows
that the FEI conjecture holds for the class of functions $f$
with bounded $\aUC^{\oplus}(f)$, and $\Inf(f) \geq 1$.
That is, for a Boolean function $f$ with $\Inf(f) \geq 1$, we have 
\begin{align*}
  \ent(\hat{f}^2) \leq 2 \cdot\aUC^{\oplus}(f)\cdot\Inf(f).
\end{align*}
We note that the reduction in \cite[Proposition E.2]{www14} shows that
removing the requirement $\Inf(f) \geq 1$ from the above inequality will prove
the FEI conjecture for all Boolean functions with $\Inf(f) \geq \log n$. 
Furthermore, if we could show the FEI conjecture for Boolean functions $f$
where $\aUC^{\oplus}(f) = \omega(1)$ is a slow-growing function of $n$,
again the padding argument in \cite{www14} shows that we would be able
to establish the FEI conjecture for all Boolean functions. 

\subsection{Further improving the bound: unambiguous $\dnf$s}
In this section, we further improve the bounds obtained
in the previous section by considering the $\dnf$ representation of $f$.
Let $\mathscr{C} := \{C_1,\ldots,C_t\}$ be an unambiguous certificate of $f$.
It covers both $1$ and $-1$ inputs of $f$.
Suppose $\{C_1,\ldots , C_{t_1}\}$ for some $t_1 < t$
is a subcube partition of $f^{-1}(-1)$ and
$\{C_{t_1 +1} , \ldots , C_t\}$ is a subcube partition of $f^{-1}(1)$.
To represent $f$, it suffices to consider $\displaystyle\bigvee_{i=1}^{t_1}C_i$.
This is a $\dnf$ representation of $f$ with an additional property
that  $\{C_1,\ldots ,C_{t_1}\}$ forms a partition of $f^{-1}(-1)$.
We call such a representation an \emph{unambiguous} $\dnf$. 
In general, a $\dnf$ representation need not satisfy this additional property.

As before, for every $C_i$, let $\tau(C_i)$ be the partial assignment
$\tau(C_i)\colon [n]\rightarrow \{-1,1,\ast\}$ corresponding to the bits
fixed by $C_i$, and $|\tau(C_i)|$ denotes the size of the set of indices
where $\tau(C_i)\neq \ast$. We know that
\begin{align*}
  \Inf(f) \leq \aUC(f,\mathscr{C}) = \sum_{i=1}^t|\tau(C_i)|\cdot 2^{-|\tau(C_i)|} = \sum_{i=1}^{t_1}|\tau(C_i)|\cdot 2^{-|\tau(C_i)|} + \sum_{i=t_1 +1}^t |\tau(C_i)|\cdot 2^{-|\tau(C_i)|}.
\end{align*}
However, the following better bound on the influence is easily seen
from the equivalence of influence and average sensitivity. 
\begin{proposition}
  \label{prop:inf-ub}
  Let $f\colon\pmset{n} \to \pmset{}$ be a Boolean function and
  $\mathscr{C}=\{C_1,\ldots ,C_t\}$ be an unambiguous certificate of $f$. Then, 
  \begin{align*}
    \Inf(f) \leq 2\cdot \min\left\{\sum_{i=1}^{t_1}|\tau(C_i)|2^{-|\tau(C_i)|}, \sum_{i=t_1 +1}^t |\tau(C_i)|2^{-|\tau(C_i)|} \right\} \leq \aUC(f,\mathscr{C}).
  \end{align*}
\end{proposition}
Therefore, it's natural to ask whether we can improve the upper bound
in Theorem~\ref{thm:partition-entropy} to
\begin{align*}
  \min\left\{\sum_{i=1}^{t_1}|\tau(C_i)|2^{-|\tau(C_i)|}, \sum_{i=t_1 +1}^t |\tau(C_i)|2^{-|\tau(C_i)|} \right\} \mbox{ (up to a constant factor)}.
\end{align*}
We remark that the quantity $\sum_{i=1}^{t_1}|\tau(C_i)|2^{-|\tau(C_i)|}$,
in a certain sense, is ``average unambiguous $1$-certificate complexity'' and,
similarly,  $ \sum_{i=t_1 +1}^t |\tau(C_i)|2^{-|\tau(C_i)|}$ captures
``average unambiguous $0$-certificate complexity''.
Moreover, answering this question is of significance to tackling
Mansour's conjecture positively.
For more details, see Section~\ref{sec:Mansour}.   

Building on our ideas from the main theorem in the previous section
and using a \emph{stronger} version of the AEP theorem
(Theorem~\ref{thm:strongaep}) we nearly establish
the aforementioned improved bound of the smaller quantity between 
``average unambiguous $1$-certificate complexity'' and
``average unambiguous $0$-certificate complexity'' on the Fourier entropy. 
Formally, we prove the following.
\begin{theorem}
  \label{thm:unambiguous-dnf}
  Let $f \colon \pmset{n} \to \pmset{}$ be a Boolean function
  and $\mathscr{C} = \{C_1, \ldots ,C_t\}$ be an unambiguous certificate
  of $f$ such that $\{C_1,\ldots , C_{t_1}\}$ for some $t_1 < t$
  is a subcube partition of $f^{-1}(-1)$ and
  $\{C_{t_1 +1} , \ldots , C_t\}$ is a subcube partition of $f^{-1}(1)$.
  Further, $ p := \Pr_x[f(x) =1]$. Then,
  \begin{align*}
  \ent(\hat{f}^2) \leq
  \begin{dcases*}
    2\left(\sum_{i=1}^{t_1} |\tau(C_i)| \cdot 2^{-|\tau(C_i)|} + p \cdot \max_{i \in \{1,\ldots ,t_1\}}|\tau(C_i)|\right), \\ 
    2\left(\sum_{i=t_1+1}^{t} |\tau(C_i)| \cdot 2^{-|\tau(C_i)|} + (1-p) \cdot \max_{i \in \{t_1+1,\ldots,t\}}|\tau(C_i)|\right). 
  \end{dcases*}
  \end{align*}
\end{theorem}
\begin{remark}
  \label{rem:one-sided-bound}
  An unsatisfactory part of the bound above is the presence of the term
  $p\cdot \max |\tau(C_i)|$.
  This is because when $\max |\tau(C_i)|$ term is \emph{not} weighted by $p$,
  it becomes a trivial bound on entropy.
  Ideally, one would like to get rid of this term altogether,
  possibly at the expense of increasing the constant factor
  in the first summand.

  Moreover, a similar bound when $\{C_1,\ldots ,C_{t_1} \}$
  is an arbitrary DNF representation of $f$
  (i.e., when the $C_i$s need not be disjoint)
  suffices to prove Mansour's conjecture (Conjecture~\ref{con:Mansour}).
  For details, see Section~\ref{sec:Mansour}.  
\end{remark}
\begin{proof}[Proof of Theorem~\ref{thm:unambiguous-dnf}]
  We only prove the inequality
  \begin{align}\label{eq:improved-inequality}  
    \ent(\hat{f}^2) \leq 2\cdot \left(\sum_{i=1}^{t_1} |\tau(C_i)| \cdot 2^{-|\tau(C_i)|} + p \cdot \max_{i \in [t_1]}|\tau(C_i)|\right).
  \end{align}
  In order to obtain the other inequality in the theorem statement,
  we can replace $f$ by $\neg f$ and observe that Fourier entropy of $f$
  and $\neg f$ is the same. We recall $\{C_1,\ldots , C_{t_1}\}$
  is a subcube partition of $f^{-1}(-1)$ and $p = \Pr_x[f(x) =1]$. 

  Like in the proof of Theorem~\ref{thm:partition-entropy},
  we define a random variable over subsets of $\pmset{n}$
  such that the support of the random variable partitions the space.
  In particular, define a random variable $\mathbf{C}$ that equals
  one of the $C_i$, $1 \leq i \leq t_1$, with probability $2^{-|\tau(C_i)|}$
  and equals the set $f^{-1}(1) := \{x \mid f(x) = 1\}$ with probability $p$.\footnote{Observe that $\sum_{i=1}^{t_1}2^{-|\tau(C_i)|}=1-p$, hence $\mathbf{C}$ is a valid random variable.}
  We consider the $M$-fold iterated product of this partition
  $\{C_1,\ldots,C_{t_1},f^{-1}(1)\}$ of $\pmset{n}$,
  which gives us a partition $\mathcal{P}$ of the space $\pmset{Mn}$
  such that $f^M$ is constant on each part.
  Overall there are $(t_1+1)^M$ parts in $\mathcal{P}$.
  For a part $\rho \in \mathcal{P}$, let $\mathds{1}_{\rho}(z)$ be
  the $\bool$-valued function that is $1$ if and only if
  $z \in \pmset{Mn}$ is in $\rho$. Furthermore, for $\rho\in \mathcal{P}$,
  let  $f^M(\rho)\in \pmset{}$ be the value of $f^M$ on the part $\rho$.
  Then, on an input $z \in \pmset{Mn}$, 
  \begin{align}
    \label{eq:folded-f-expansion}
    f^M(z)  = \sum_{\rho \in \mathcal{P}} f^M(\rho) \cdot \mathds{1}_\rho(z) 
    = \sum_{\rho \in T^{\ast (M)}_\varepsilon(\mathbf{C})} f^M(\rho) \cdot \mathds{1}_\rho(z) + \sum_{\rho \not\in T^{\ast(M)}_\varepsilon(\mathbf{C})} f^M(\rho) \cdot \mathds{1}_\rho(z).
  \end{align}
  where $T_{\varepsilon}^{\ast (M)}(\mathbf{C})$ is the strongly typical set
  with respect to the distribution of $\mathbf{C}$ and $\varepsilon > 0$.
  Similar to the proof of Theorem~\ref{thm:partition-entropy},
  we now define a set $\mathcal{B}$ of Fourier coefficients of $f^M$,
  which we argue below to be of ``small'' size and ``large'' Fourier weight. 
  \begin{align}
    \mathcal{B} :=\Big\{ (S_1, \ldots ,S_M) \subseteq [n]^M \mid
    \fc{\mathds{1}_\rho}{S_1\cup\cdots \cup S_M} \neq 0 \mbox{ for some }\rho \in T^{\ast (M)}_\varepsilon(\mathbf{C})  \Big\}.
    \end{align}
   \begin{claim}
    \label{claim:dnf-good-sets}
    \begin{align*}
    \left|\mathcal{B}\right| \leq  2^{2M \delta}\cdot 2^{2M\left(\sum_{i=1}^{t_1} |\tau(C_i)| (2^{-| \tau(C_i)|} + \varepsilon') ~+~  (p +\varepsilon' ) \max_{i=1}^{t_1} |\tau(C_i)| \right)},
    \end{align*}
    where $\varepsilon' = \varepsilon/(t_1 +1)$ and $\delta >0$ is such that
    $\delta \to 0$ as $\varepsilon \to 0$. 
  \end{claim}
  \begin{proof}[Proof of Claim]
    Similar to Claim~\ref{claim:good-sets} we first bound
    the size of the strongly typical set $T^{\ast (M)}_\varepsilon(\mathbf{C})$
    and then count contributions of a strongly typical sequence
    $\rho$ to $\mathcal{B}$. The contribution of $\rho$ can
    be upper bounded by the sparsity of the Fourier expansion of
    the indicator function $\mathds{1}_\rho$. Therefore, 
    \begin{align*}
      |\mathcal{B}| \leq |T^{\ast (M)}_\varepsilon(\mathbf{C})| \cdot \max_{\rho \in \mathcal{P}}\{\mbox{Fourier sparsity of }\mathds{1}_\rho\}.
    \end{align*}
    By the strong AEP theorem, Theorem~\ref{thm:strongaep}~\ref{sts-iii},
    we have \(|T^{\ast (M)}_\varepsilon(\mathbf{C})| \leq 2^{M (\mathbb{H}(\mathbf{C}) + \delta)}\). We now bound the Fourier sparsity of $\mathds{1}_\rho$
    for a strongly typical $\rho$. By Parseval's Theorem
    (Fact~\ref{fact:plancherel}), we have
    \begin{align*} 
      \sum_{(S_1,\ldots ,S_M)} \fcsq{\mathds{1}_\rho}{S_1\cup\cdots \cup S_M} = \Pr_z\left[\mathds{1}_\rho(z) = 1\right]. 
    \end{align*}
    Therefore, 
    \begin{align}
      \label{eq:sparsitypperboundbystrongtypical}
      \mathsf{sparsity}(\widehat{\mathds{1}_\rho}) \cdot \min_{(S_1,\ldots ,S_M)}\fcsq{\mathds{1}_\rho}{S_1\cup\cdots \cup S_M} \leq  \Pr_z\left[\mathds{1}_\rho(z) = 1\right] \leq 2^{-M(\mathbb{H}(\mathbf{C}) - \delta)}, 
    \end{align}
    where the second inequality follows from the strong AEP theorem,
    Theorem~\ref{thm:strongaep}~\ref{sts-ii}, and
    $\mathsf{sparsity}(\widehat{\mathds{1}_\rho})$ denotes
    the Fourier sparsity of $\mathds{1}_\rho$.
    To obtain an upper bound on the sparsity we now establish
    a lower bound on the magnitude of the non-zero Fourier coefficients
    of $\mathds{1}_\rho$.

    Since $\rho \in \{C_1, \ldots , C_{t_1},f^{-1}(1)\}^M$,
    the Fourier expansion of the indicator function $\mathds{1}_\rho$
    is a product of Fourier expansion of indicator functions of
    $C_1, \ldots , C_{t_1}$ and $f^{-1}(1)$. The indicator function of
    $C_i$ is just an $\band$ over the appropriate subset of literals,
    and the indicator function of $f^{-1}(1)$ is
    $\displaystyle 1 - \sum_{i=1}^{t_1} \mathds{1}_{C_i}$.
    Hence, a lower bound on the magnitude of
    a non-zero Fourier coefficient of $\mathds{1}_\rho$ is given by
    \begin{align*}
      \prod_{i=1}^{t_1} \left(\frac{1}{2^{|\tau(C_i)|}}\right)^{N(C_i;\rho)} \cdot \left(\frac{1}{2^{\max_{i \in \{1, \ldots ,t_1\}} \left\{|\tau(C_i)|\right\}}}\right)^{N(f^{-1}(1);\rho)},
    \end{align*}
    where $N(C_i ; \rho)$ (respectively, $N(f^{-1}(1);\rho)$)
    is the number of occurrences of $C_i$ (respectively, $f^{-1}(1)$)
    in $\rho$. Therefore, squaring the above lower bound and
    using it in Eq.~\eqref{eq:sparsitypperboundbystrongtypical} we obtain
    \begin{align}
      \label{eq:sparsity-ub}
      \mathsf{sparsity}(\widehat{\mathds{1}_\rho}) & \leq \frac{1}{2^{M(\mathbb{H}(\mathbf{C}) - \delta)}}\cdot \prod_{i=1}^{t_1} 2^{2 |\tau(C_i)| \cdot N(C_i;\rho)} \cdot 2^{2\max_{i \in [t_1]}\left\{ |\tau(C_i)| \right\}\cdot N(f^{-1}(1);\rho) }, \nonumber \\
      & = \frac{1}{2^{M(\mathbb{H}(\mathbf{C}) - \delta)}}\cdot 2^{2 \cdot \left(\sum_{i=1}^{t_1}|\tau(C_i)|\cdot N(C_i;\rho) ~+~ \max_{i\in [t_1]}\left\{|\tau(C_i)|\right\} \cdot N(f^{-1}(1);\rho) \right)}, \nonumber \\
      & \leq \frac{1}{2^{M(\mathbb{H}(\mathbf{C}) - \delta)}}\cdot 2^{2M\cdot \left(\sum_{i=1}^{t_1}|\tau(C_i)|\cdot \left(2^{-|\tau(C_i)|} + \varepsilon' \right) ~+~ \max_{i\in [t_1]}\left\{|\tau(C_i)| \right\} \cdot \left(p + \varepsilon'\right) \right) },
    \end{align}
    where the last inequality follows from the strong AEP theorem
    and $\varepsilon' = \varepsilon/(t_1+1)$.
    Now using inequality~\eqref{eq:sparsity-ub} along with the bound
    on the size of the strongly typical set
    $T^{\ast (M)}_\varepsilon(\mathbf{C})$ we obtain the claimed bound
    on the size of $\mathcal{B}$, thereby completing the proof of the claim.
  \end{proof}
  The next claim shows that the Fourier coefficients \emph{not} in
  $\mathcal{B}$ have low total Fourier weight.
  \begin{claim}
    \label{claim:dnf-bad-sets}
    \begin{align*}
      \sum_{(S_1,S_2,\ldots,S_M) \not\in \mathcal{B}} 
      \fcsq{f^M}{S_1\cup\cdots \cup S_M} \leq \varepsilon .
    \end{align*}
  \end{claim}
  \begin{proof}[Proof of Claim]
    We omit the proof as it is similar to the proof of Claim~\ref{claim:bad-sets}. 
  \end{proof}
  Now bounding the Fourier entropy of $f^M$ as in the proof of
  Theorem~\ref{thm:partition-entropy} and taking the limit as $M \to \infty$,
  we obtain the following bound on the Fourier entropy of $f$,
  \begin{align*}
    \ent(\hat{f}^2) \leq 2\cdot \left(\sum_{i=1}^{t_1} |\tau(C_i)| \cdot 2^{-|\tau(C_i)|} + p \cdot \max_{i \in [t_1]}|\tau(C_i)|\right).
  \end{align*}
\end{proof}

\subsubsection{Extension to affine subspace partition}
Since our techniques are oblivious to the change of basis,
again analogous to Theorem~\ref{thm:subspace-partition-entropy}
we obtain the following generalization of Theorem~\ref{thm:unambiguous-dnf}
to the setting when $\{C_1, \ldots ,C_t\}$ forms
a monochromatic affine subspace partition.
We state the generalization below without the proof. 
\begin{theorem}[Restatement of Theorem~\ref{thm:one-sided-subspace-partition-intro}]
  \label{thm:one-sided-subspace-partition}
  Let $f\colon\pmset{n} \to \pmset{}$ be a Boolean function and
  $\mathscr{C} = \{C_1, \ldots ,C_t\}$ be a monochromatic
  affine subspace partition of $\pmset{n}$ with respect to $f$
  such that $\{C_1,\ldots , C_{t_1}\}$ for some $t_1 < t$ is an affine subspace
  partition of $f^{-1}(-1)$ and $\{C_{t_1 +1} , \ldots , C_t\}$ is
  an affine subspace partition of $f^{-1}(1)$. Further, $ p := \Pr_x[f(x) =1]$. 
  Then,
  \begin{align*}
    \ent(\hat{f}^2) \leq
    \begin{dcases*}
     2\left(\sum_{i=1}^{t_1} \mathsf{co\text{-}dim}(C_i) \cdot 2^{-\mathsf{co\text{-}dim}(C_i)} + p \cdot \max_{i \in \{1,\ldots ,t_1\}}\mathsf{co\text{-}dim}(C_i)\right), \\ 
     2\left(\sum_{i=t_1+1}^{t} \mathsf{co\text{-}dim}(C_i) \cdot 2^{-\mathsf{co\text{-}dim}(C_i)} + (1-p) \cdot \max_{i \in \{t_1+1,\ldots ,t\}}\mathsf{co\text{-}dim}(C_i)\right).
   \end{dcases*}
  \end{align*}  
\end{theorem}

\subsection{Discussions on certificate complexity and Mansour's conjecture}
\label{sec:Mansour}
An important consequence of the FEI conjecture, among many,
is a positive answer to the long-standing conjecture of Mansour.  
\begin{conjecture}[Mansour's Conjecture~\cite{Mansour94}]
  \label{con:Mansour}
  Let $f \colon \pmset{n}\rightarrow \pmset{}$ be a Boolean function that is
  representable by a $t$-term $\dnf$. For every constant $\varepsilon > 0$,
  there exists a polynomial $p$ over $\pmset{}$ with sparsity
  $\mathsf{poly}(t)$ such that
  \(\E_x\left[(f(x)-p(x))^2\right] \leq \varepsilon\).
  (The exponent in $\mathsf{poly}(t)$ can depend on $1/\varepsilon$.) 
\end{conjecture}
In fact, Mansour's original conjecture states that sparsity of
the polynomial $p$ (in the conjecture above) can be taken to be
$t^{O\left(\log \frac{1}{\varepsilon}\right)}$.
Mansour's conjecture has a number of important consequences.
For instance, Gopalan et al.~\cite{GKK08} showed that
a positive answer to Mansour's conjecture (Conjecture~\ref{con:Mansour})
would imply that $\dnf$ formulas can be agnostically learned
in polynomial time up to any constant error parameter.
This has been a long-standing open question~\cite{gopalan:DNF}
in computational learning theory.

Earlier we saw that $\ent(\hat{f}^2)\leq O(\aUC^\oplus(f))$.
In fact we saw a somewhat stronger bound
(Theorem~\ref{thm:one-sided-subspace-partition}); informally,
capturing ``average unambiguous $1$-parity-certificate complexity''.
An interesting follow-up question in light of such results,
in particular Theorems~\ref{thm:partition-entropy}
and~\ref{thm:unambiguous-dnf}, is if one could strengthen
these upper bound to $O(\min\{\certificate^{0}(f),\certificate^{1}(f)\})$. 
In this section we observe that this bound on the Fourier-entropy
in terms of $\certificate^0(f),\certificate^1(f)$
(which is clearly weaker than the FEI conjecture)
suffices to establish Mansour's conjecture.

We remark that it was implicit in previous works~\cite{Kalai-blog,gopalan:DNF,KLW10,OWZ11} that one doesn't need the full power of the FEI conjecture
to establish Mansour's conjecture. However, the following
question$\colon$ \emph{What is the weakest form of the FEI conjecture that still implies Mansour's conjecture?},
was left unexplored. Our observation sharpens this relationship and
establishes Mansour's conjecture as a natural next step towards
resolving the FEI conjecture. 

We now formally state the weaker conjecture than the FEI conjecture
that suffices to imply Mansour's conjecture.
\begin{conjecture}
  \label{con:ent-vs-01-certificate}
  There exists a universal constant $\lambda > 0$ such that for every Boolean function $f\colon\pmset{n}\rightarrow \pmset{}$, we have
  \begin{align*}
    \ent(\hat{f}^2) \leq \lambda\cdot \min\left\{\certificate^{0}(f),\certificate^{1}(f)\right\}.
  \end{align*}
\end{conjecture}
It is weaker than the FEI conjecture because
\(\Inf(f) \leq \min\left\{\certificate^{0}(f),\certificate^{1}(f)\right\}\)~\cite{Bop97,Traxler09,Amano11}. 
To establish the implication we will use the following equivalent
form of Conjecture~\ref{con:ent-vs-01-certificate}.
\begin{conjecture}
  \label{con:ent-vs-1-certificate}
  There exists a universal constant $\lambda > 0$ such that for every
  Boolean function $f\colon\pmset{n}\rightarrow \pmset{}$, we have
  \begin{align*}
    \ent(\hat{f}^2) \leq \lambda\cdot \certificate^{1}(f).
  \end{align*}
\end{conjecture}
Before establishing the implication, we quickly argue that
Conjectures~\ref{con:ent-vs-01-certificate} and~\ref{con:ent-vs-1-certificate}
are equivalent.
It is easily seen that Conjecture~\ref{con:ent-vs-01-certificate} implies
Conjecture~\ref{con:ent-vs-1-certificate}. For the reverse direction, note that
both $f$ and $\neg f$ have the same Fourier-entropy,
while $\certificate^0(f)$ and $\certificate^1(f)$ reverse roles.  

We also note that Conjecture~\ref{con:ent-vs-1-certificate} is
readily seen to imply the weaker version of
``Mansour's conjecture for width'' posed by \cite{GMR13} which in turn
is known to imply Conjecture~\ref{con:Mansour}. 
However, for the sake of completeness, we present here
a direct argument to verify this implication.
It is also worthwhile to mention the recent works on DNF sparsification
\cite{LZ19,LWZ20} that can be used to establish the reverse implication too
(see \cite[Lemma 6.4]{GMR13}).  
\begin{proposition}
  \label{prop:ent-vs-cert-vs-Mansour}
  Conjecture~\ref{con:ent-vs-1-certificate} implies
  Conjecture~\ref{con:Mansour}. 
\end{proposition}
\begin{proof}
  Let $f$ be a $t$-term $\dnf$ and suppose $\delta_1,\delta_2> 0$
  are constants which we pick later.
  Let $g$ be a Boolean function obtained from $f$ by dropping all terms
  of length more than $\log (4t/\delta_1)$ in the $\dnf$ for $f$.
  Over the uniform distribution, each term of length greater than
  $\log (4t/\delta_1)$ equals $1$ with probability at most $\delta_1/4t$.
  Then $g(x)$ and $f(x)$ differ only if $x$ is accepted by a term of length
  greater than $\log (4t/\delta_1)$.
  Since there are at most $t$ terms, by a union bound, we get
  \begin{align}
    \label{eq:f-gdifference}
    \E_x\left[(f(x)-g(x))^2\right] \leq 4\cdot t\cdot \frac{\delta_1}{4t}=\delta_1.
  \end{align}
  Using Conjecture~\ref{con:ent-vs-1-certificate} for the function $g$, we get
  $\ent(\hat{g}^2) \leq \lambda \cdot \certificate^1(g) \leq \lambda \log (4t/\delta_1)$.
  We now construct a polynomial $p$ by defining its Fourier coefficients
  as follows:
  \begin{align*}
  \widehat{p}(S)= \begin{cases} 
    \widehat{g}(S) & \text{ if }|\widehat{g}(S)|\geq 2^{-\ent(\hat{g}^2)/(2\delta_2)}, \\
    0 & \text{ otherwise. } 
  \end{cases}
  \end{align*}
  By Parseval's identity (Fact~\ref{fact:plancherel}),
  it follows that the number of non-zero Fourier coefficients in $p$
  is at most $2^{\ent(\hat{g}^2)/\delta_2}$. Additionally we have that
  \begin{align}
    \label{eq:g-pdifference}
    \E_x\left[(g(x)-p(x))^2\right]=\sum_S \left(\fc{g}{S}-\fc{p}{S}\right)^2=\sum_{S\colon |\fc{g}{S}|<2^{-\ent(\hat{g}^2)/(2\delta_2)}} \fcsq{g}{S} \leq \delta_2,
  \end{align}
  where the inequality follows from the Markov inequality.  

  Using traingle inequality on $\ell_2$-norm,
  we can easily bound $\E[(f-p)^2]$.
  \begin{align*}
    \E\left[(f-p)^2\right] = \E\left[(f-g+g-p)^2\right] \leq (\sqrt{\delta_1} + \sqrt{\delta_2})^2, 
  \end{align*}
  where the inequality used substitutions from
  Eq.~\eqref{eq:f-gdifference} and~\eqref{eq:g-pdifference}. 
  By picking $\delta_1 = \delta_2 = \varepsilon/4$ we get
  $\E[(f-p)^2] \leq \varepsilon$, which ensures that $p$ has
  the approximation needed for Mansour's conjecture.
  Additionally, the Fourier sparsity of~$p$ is at most 
  \begin{align*}
    2^{\ent(\hat{g}^2)/\delta_2}\leq  2^{\lambda \log(4t/\delta_1)/\delta_2}=\Big(\frac{16t}{\varepsilon}\Big)^{\frac{4\lambda}{\varepsilon}}.
  \end{align*}
\end{proof}
We end this section with another open problem that could form an 
intermediate step towards resolving Mansour's conjecture. 
The following seemingly weaker conjecture than
Conjecture~\ref{con:ent-vs-01-certificate} is not known to imply
Mansour's conjecture. 
\begin{conjecture} 
  \label{con:ent-vs-certificate} 
  There exists a universal constant $\lambda > 0$ such that for any
  Boolean function $f\colon\pmset{n}\rightarrow \pmset{}$,
  \begin{align*}
    \ent(\hat{f}^2) \leq \lambda\cdot \max\left\{\certificate^{0}(f),\certificate^{1}(f)\right\}=\lambda \cdot \certificate(f).
  \end{align*}
\end{conjecture}

\section{Better bounds on Fourier Min-entropy}
\label{sec:minentropy}
The \emph{Fourier Min-entropy-Influence conjecture} (FMEI)
is a natural weakening of the FEI conjecture that
has received much less attention compared to the FEI conjecture.
The FMEI conjecture was raised by O'Donnell and others
in~\cite{OWZ11,ODonnell-book} as a simpler question to tackle,
given the hardness of resolving the FEI conjecture.
We restate the FMEI conjecture below. 
\begin{conjecture}[FMEI conjecture]
  There exists a universal constant $C>0$ such that
  for every Boolean function $f\colon\pmset{n}\rightarrow \pmset{}$,
  we have $\ent_{\infty}(\hat{f}^2)\leq  C \cdot \Inf(f)$.
\end{conjecture}
Although the FMEI conjecture is a natural first step towards
proving the FEI conjecture, it is also interesting in its own right.
For example, it is known to imply the famous KKL theorem~\cite{KKL88}. 
In fact, until very recently\footnote{Eldan and Gross \cite{EG20} developed new techniques using stochastic analysis that give a new way to prove the KKL theorem among other things.}, 
we did not know of any proof of the KKL theorem
that didn't go through hypercontractivity or
logarithmic Sobolev inequalities, which makes proving
the FMEI conjecture even more interesting.
The KKL theorem states that for every Boolean function
$f\colon\pmset{n}\rightarrow \pmset{}$,
there exists an index $i\in [n]$ such that
\(\Inf_i(f)\geq \Var(f) \cdot \Omega\left(\frac{\log n}{n}\right)\).
We now present an argument of \cite{OWZ11} that shows
how a positive answer to the FMEI conjecture
implies the KKL theorem for balanced functions: assume that $f$ is balanced
(i.e., $\fc{f}{\emptyset}=\E_x[f(x)]=0$).
Then the FMEI conjecture implies the existence of
$\emptyset \neq T\subseteq [n]$ such that
$\fcsq{f}{T}\geq 2^{-C\cdot \Inf(f)}$. Furthermore, for every $i\in T$, we have 
\begin{align}
  \label{eq:minfeikkl}
  \Inf_i(f)=\sum_{S:S\ni i}\fcsq{f}{S} \geq \fcsq{f}{T}\geq 2^{-C\cdot \Inf(f)}\geq 2^{-C n\cdot \max_{j} \left\{\Inf_j(f)\right\}},
\end{align}
where the first inequality follows because $T$ contains $i$,
the second inequality follows from the FMEI conjecture and
the last inequality because $\Inf(f)\leq n\cdot\max_{j}\{\Inf_j(f)\}$.
However note that $\max_j\{\Inf_j(f)\}$ clearly upper bounds
the left-hand side of Eq.~\eqref{eq:minfeikkl}. 
Thus, we have
\begin{align*}
  \max_{j \in [n]}\left\{\Inf_j(f)\right\}\geq 2^{-C n\cdot \max_{j} \left\{\Inf_j(f)\right\}}.
\end{align*}
Rearranging this inequality, we obtain
\(\max_{j\in [n]}\{\Inf_j(f)\} \geq \Omega\left(\frac{\log n}{n}\right)\),
which is the KKL theorem for balanced functions.
The proof can also be extended to non-balanced functions though using the FEI conjecture (see \cite{OWZ11}). 

We now prove Theorem~\ref{thm:minentropyupperb},
which is our main contribution in this section.
In the following theorem, we give upper bounds on
$\ent_{\infty}(\hat{f}^2)$ in terms of analytic and
combinatorial measures of Boolean functions. 
\begin{theorem}[Restatement of Theorem~\ref{thm:minentropyupperb-intro}]
  \label{thm:minentropyupperb}
  Let $f\colon\pmset{n}\rightarrow \pmset{}$ be a Boolean function. Then,  
  \begin{enumerate}[label=(\roman*)]
  \item\label{min-vs-approx-norm} For every $\varepsilon \geq 0$,
    \(\ent_{\infty}(\hat{f}^2)\leq 2\cdot\log \left(\|\widehat{f}\|_{1,\varepsilon}/(1-\varepsilon)\right)\).
  \item\label{min-vs-cert} \(\ent_{\infty}(\hat{f}^2)\leq  2\cdot \certificate_{\min}^{\oplus}(f)\). 
  \item\label{min-vs-rpdt} \(\ent_{\infty}(\hat{f}^2)\leq  2(1+\log_2 3)\cdot R_2^\oplus(f)\). 
  \end{enumerate}
\end{theorem}
Before giving a proof, we first make a few remarks about
the second statement in the theorem above. 
The FMEI conjecture asks if
\(\ent_{\infty}(\hat{f}^2)\leq C\cdot \as(f)\)?
Since we also know that for every $x\in \pmset{n}$,
we have $\s(f,x)\leq \certificate(f,x)$,
a weaker question than the FMEI conjecture
(with a larger right-hand side) would be,
is \(\ent_{\infty}(\hat{f}^2)\leq C\cdot \ac(f)\)?
In the theorem above, we give a positive answer to this question
and in fact show that
\(\ent_{\infty}(\hat{f}^2)\leq  2\cdot\certificate_{\min}^{\oplus}(f)\).
Observe that $\certificate_{\min}^{\oplus}(f)$ is
not only a lower bound on $\ac^{\oplus}(f)$ (and in turn $\ac(f)$),
but it is the parity certificate complexity on the ``easiest'' input to $f$,
unlike $\certificate^{\oplus}(f)$ where
the complexity is measured according to the ``hardest'' input $x$ to $f$.
In our perspective, this brings us closer to proving the FMEI conjecture.
In fact, we identify a non-trivial class of Boolean functions
for which $\certificate_{\min}(f)$ lower bounds $\Inf(f)$,
and hence establish the FMEI conjecture for
this class (Theorem~\ref{thm:meic-dnf}). 
\begin{proof}[Proof of Theorem~\ref{thm:minentropyupperb}]
  We prove the three parts separately as follows.
  \paragraph{Part~\ref{min-vs-approx-norm}}
  Fix $\varepsilon \geq 0$ and $d\in [n]$.
  Given a Boolean function $f$,
  suppose $p$ is a degree-$d$ polynomial that minimizes
  \begin{align*}
    \|\widehat{f}\|_{1,\varepsilon,d}=\min \left\{\|\widehat{p}\|_1:  \deg(p)\leq d \text{ and } |p(x)-f(x)| \leq \varepsilon \text{ for every } x\in \pmset{n} \right\},
  \end{align*}
  where the minimization is over all polynomials.
  Alternatively, $\|\widehat{f}\|_{1,\varepsilon,d}$
  can also be expressed as the following linear program
  and $p$ minimizes this program.
  \begin{align*}
    \boxed{
      \begin{array}{lll@{}ll}
	\|\widehat{f}\|_{1,\varepsilon,d}=& \min  &\sum_{S}|c_S| & \\
	&\text{subject to}& \Big| f(x)- \sum_{S : |S|\leq d} c_S\chi_{S}(x) \Big|\leq \varepsilon & \quad \text{ for every } ~x\in \pmset{n}  &\\
	&          &      c_S\in \reals                                          &\quad \text{ for every } S:|S| \leq d  & 
      \end{array}
    }
  \end{align*}
  Note that for every $\varepsilon \geq 0$ and
  $d \geq {\deg}_{\varepsilon}(f)$
  the above linear program is feasible.
  From standard manipulations, the dual of the linear program is as follows. 
  \begin{align*}
    \boxed{
      \begin{array}{lll@{}ll}
	& \max  &  \sum_{x\in \pmset{n}}\phi(x)f(x)-\varepsilon\sum_{x \in \pmset{n}}|\phi(x)| &\\
	&\text{subject to}& | \fc{\phi}{S}| \leq \frac{1}{2^n} & \quad  \text{ for every } S:|S| \leq d   &\\
	&          &      \phi(x)\in \reals                                          &\quad  \text{ for every } x\in \pmset{n}  & 
      \end{array}
    }
  \end{align*}
  Observe that both linear programs are feasible for
  $d \geq {\deg}_{\varepsilon}(f)$.
  Therefore, from the duality theorem of linear programs,
  the objective value of any dual feasible solution lower bounds
  the primal optimum and, moreover,
  the two programs have the same optimum value.
  We thus obtain the following characterization of
  $\|\widehat{f}\|_{1,\varepsilon,d}$.\footnote{We remark that similar linear program characterizations of approximate degree of Boolean functions have appeared before in the works of Sherstov~\cite{sherstov:pmm} and Bun and Thaler~\cite{bun&thaler:dualpoly}.}
  \begin{lemma}
    \label{lem:duality}
    Let $f\colon\pmset{n}\to\pmset{}$, $\varepsilon \geq 0$,
    and $d\in [n]$ such that $d \geq {\deg}_{\varepsilon}(f)$.
    Then, $\|\widehat{f}\|_{1,\varepsilon,d}\geq T$ if and only if
    there exists a polynomial $\phi\colon\pmset{n}\to \reals$
    satisfying $|\fc{\phi}{S}|\leq 2^{-n}$ for all $|S|\leq d$ and
    \begin{align*}
      \sum_{x\in \pmset{n}}\phi(x)f(x)-\varepsilon\sum_{x \in \pmset{n}}|\phi(x)| \geq T.
    \end{align*}
  \end{lemma}
  Let us consider $\phi(x) = \frac{f(x)}{2^n\max_S |\fc{f}{S}|}$.
  Clearly the dual constraints are satisfied,
  and the objective value for this choice of $\phi$ is given by
  \begin{align*}
    \sum_{x\in \pmset{n}}\phi(x)f(x)-\varepsilon\sum_{x \in \pmset{n}}|\phi(x)|  = \frac{1-\varepsilon}{\max_S|\fc{f}{S}|}.
  \end{align*}
  The equality holds since $\phi(x)f(x) = |\phi(x)|$.
  Now, by Lemma~\ref{lem:duality}, we have   
  \begin{align*}
    \|\widehat{f}\|_{1,\varepsilon,d}\geq \frac{1-\varepsilon}{\max_S|\fc{f}{S}|}.
  \end{align*}
  Therefore, we obtain,
  \begin{align*}
    \ent_\infty(\hat{f}^2) \leq 2\cdot \log \left(\frac{\|\widehat{f}\|_{1,\varepsilon,d}}{1-\varepsilon}\right).
  \end{align*}
  Since $d$ is arbitrary, in order to ensure feasibility of
  the program we pick $d= {\deg}_{\varepsilon}(f)$.
  The first part of the theorem follows since
  \(\|\widehat{f}\|_{1,\varepsilon,n}=\|\widehat{f}\|_{1,\varepsilon}\).
  \paragraph{Part~\ref{min-vs-cert}}
  Suppose $\certificate_{\min}^{\oplus}(f)=k$. By definition of
  $\certificate_{\min}^{\oplus}(f)$, there exists an affine subspace
  $H \subseteq \pmset{n}$ such that $\mathsf{co\text{-}dim}(H)=k$ and
  $f$ is constant on $H$. Without loss of generality,
  assume that $f(x) = -1$ for every $x\in H$.
  Since $\mathsf{co\text{-}dim}(H)$ equals $k$,
  $H$ is given by a set of $k$ (linearly independent) parity constraints.
  That is, there exist $k$ linearly independent vectors
  $S_1,\ldots ,S_k \in \bool^n$, and $b_1,\dots,b_k \in \sbool$, such that 
  \begin{align*}
    H = \biggl\{ x \in \pmset{n} \colon \text{ for every } j\in [k] , \prod_{i \in \supp(S_j)}x_i = b_j \biggr\}.
  \end{align*}
  Consider the indicator function
  $\mathds{1}_{H} \colon\pmset{n}\rightarrow \pmset{}$,
  which evaluates to $-1$ for every $x\in H$ and $1$ otherwise.
  The Fourier expansion of $\mathds{1}_H$ is easy to understand.
  Observe that $H$ can be viewed as an $\band$ over parities or
  negated-parities. 
  For $j\in [k]$, let $y_j = \prod_{i \in \supp(S_j)} x_i$.
  It is now easily seen that 
  \begin{align}
    \label{eq:writingHperp}
    \mathds{1}_{H}(x) = \band (-b_1y_1,-b_2y_2,\ldots ,-b_ky_k).
  \end{align}
  Recall that $b_j$ is fixed, thus $-b_jy_j$ is either $y_j$ or $-y_j$. 
  Writing out the Fourier expansion for the $\band$ function in
  Eq.~\eqref{eq:writingHperp}, it follows that
  \begin{align*}
    \mathds{1}_{H}(x) &= \left(1 - \frac{1}{2^{k-1}}\right) + \sum_{T \subseteq [k]\colon T \neq \emptyset} \frac{(-1)^{|T|+1}}{2^{k-1}} \prod_{j \in T}-b_jy_j\\
    &=\left(1 - \frac{1}{2^{k-1}}\right) + \sum_{T \subseteq [k]\colon T \neq \emptyset} \left(\frac{-\prod_{j\in T}b_j}{2^{k-1}}\right) \prod_{j \in T}\left( \prod_{i \in \supp(S_j)}x_i\right).
  \end{align*}
  Now using the fact that $x_i^2 =1$, observe that the monomial
  $\prod_{j \in T}\prod_{i \in \supp(S_j)}x_i$ simplifies to a multilinear monomial.
  We further observe that for each non-empty $T$,
  we can simplify $\prod_{j \in T}\prod_{i \in \supp(S_j)}x_i$
  to a distinct multilinear monomial. This is a consequence of
  linear independence of $S_1,S_2, \ldots ,S_k$. Let us denote the set of
  non-zero Fourier coefficients of $\mathds{1}_H$ by $\mathcal{T}$.
  By what we argued just now, it follows that $|\mathcal{T}| = 2^k$. 
  We are now ready to conclude the proof.
  \begin{lemma}\label{lem:1-certificate}
    There exists a set $T \subseteq [n]$ such that
    $|\fc{f}{T}| \geq \frac{1}{2^k}$ . 
  \end{lemma}
  \begin{proof}
    Let \(\Pr_x[f(x) = -1] = p\).
    We consider the correlation between $f$ and $\mathds{1}_H$.
    \begin{align}
      \label{eq:via-counting}
      \langle f , \mathds{1}_H \rangle = \E_x\left[f(x)\cdot \mathds{1}_H(x)\right] =\frac{1}{2^k} +(-1) \left(p-\frac{1}{2^k}\right) + (1-p) = (1-2p) + \frac{1}{2^{k-1}} = \fc{f}{\emptyset} + \frac{1}{2^{k-1}}. 
    \end{align}
    On the other hand,
    \begin{align}\label{eq:via-parseval}
      \langle f, \mathds{1}_H \rangle = \sum_{S \subseteq [n]}\fc{f}{S}\cdot\fc{\mathds{1}_H}{S}
      = \fc{f}{\emptyset}\cdot\left(1-\frac{1}{2^{k-1}} \right) + \sum_{T \in \mathcal{T}\colon T \neq \emptyset} \fc{f}{T}\cdot\left(\frac{-\prod_{j\in T}b_j}{2^{k-1}}\right).
    \end{align}
    Putting together Eq.~\eqref{eq:via-counting} and
    Eq.~\eqref{eq:via-parseval} we have,
    \begin{align}\label{eq:linear-relation}
      1 = -\fc{f}{\emptyset} + \sum_{T \in\mathcal{T} \colon T \neq \emptyset} \fc{f}{T}\cdot\left(-\Pi_{j \in T}b_j\right) \leq \sum_{T \in \mathcal{T}} |\fc{f}{T}|.   
    \end{align}
    Since $|\mathcal{T}| = 2^k$,
    we obtain \(\max_{T \in \mathcal{T}} |\fc{f}{T}| \geq \frac{1}{2^k}\). 
  \end{proof}
  This completes the proof of part \ref{min-vs-cert}.
  \begin{remark}
    \label{rem:min-vs-cmin}
    We note that in fact the proof of part~\ref{min-vs-cert},
    in particular Eq.~\eqref{eq:linear-relation}, 
    shows that the following bounds hold when $H$ is a subcube.
    \begin{enumerate}[label=(\alph*)]
    \item There exists a non-empty set $S\subseteq [n]$ of size at most
      $\certificate_{\min}(f)$ such that 
      \begin{align*}
        \left|\fc{f}{S}\right| \geq \frac{1 - |\fc{f}{\emptyset}|}{2^{\certificate_{\min}(f)} - 1} \geq \frac{\Var(f)}{2^{\certificate_{\min}(f) + 1}}.
      \end{align*}
    \item There exists a non-empty set $S\subseteq [n]$ of size at most
      \(\max\left\{\certificate^0_{\min}(f), \certificate^1_{\min}(f)\right\}\)
      such that 
      \begin{align*}
        \left|\fc{f}{S}\right|
        \geq \frac{1 + |\fc{f}{\emptyset}|}{2^{\max\left\{\certificate^0_{\min}(f), \certificate^1_{\min}(f)\right\}} - 1}
        \geq \frac{\sqrt{\Var(f)}}{2^{\max\left\{\certificate^0_{\min}(f), \certificate^1_{\min}(f)\right\}}}.
      \end{align*}
      This inequality follows from using either $f$ or $-f$ in
      Eq.~\eqref{eq:linear-relation},
      depending on whether $\fc{f}{\emptyset}$ is positive or not.  
    \end{enumerate}
    The size bound on $S$ follows because when $H$ is a subcube
    the Fourier expansion of $\mathds{1}_{H}$ is  supported
    on monomials of degree at most $\mathsf{co\text{-}dim}(H)$.
    We end the remark here and continue with the proof of the theorem. 
  \end{remark}
  
  \paragraph{Part~\ref{min-vs-rpdt}}
  Consider a randomized parity-decision tree $R_\mu$ computing $f$
  with probability at least $2/3$.
  Let $\mathcal{T}$ be the set of deterministic parity-decision trees
  such that $R_\mu$ assigns a non-zero probability to every
  $T\in \mathcal{T}$. By definition,
  it then follows that $\E_x\E_{T\sim \mu} [f(x)T(x)]\geq 1/3$.
  This also shows that
  \begin{align}\label{eq:r2-lb}
    \frac{1}{3} \leq \E_x\E_{T\sim \mu} [f(x)T(x)] = \E_{T \sim \mu}\E_x[f(x)T(x)]
    = \E_{T\sim \mu}\left[\sum_{S\subseteq [n]}\fc{f}{S}\fc{T}{S} \right]. 
  \end{align}
  On the other hand, one can upper bound the last expression
  in Eq.~\eqref{eq:r2-lb} as follows
  \begin{align}\label{eq:r2-ub}
    \E_{T\sim \mu}\left[\sum_{S\subseteq [n]}\fc{f}{S}\fc{T}{S} \right]
    \leq \E_{T\sim \mu}\left[\sum_{S\subseteq [n]}|\fc{f}{S}||\fc{T}{S}| \right]
    \leq \left(\max_{S\subseteq [n]}|\fc{f}{S}|\right) \E_{T\sim \mu}\left[ \sum_{S\subseteq [n]}|\fc{T}{S}|\right].
  \end{align}
  Putting together Eq.~\eqref{eq:r2-lb} and Eq.~\eqref{eq:r2-ub} we have,
  \begin{align*}
    \frac{1}{3} \leq \left(\max_{S\subseteq [n]}|\fc{f}{S}|\right) \E_{T\sim \mu}\left[ \sum_{S\subseteq [n]}|\fc{T}{S}|\right] \leq \left(\max_{S\subseteq [n]}|\fc{f}{S}|\right) 2^{R_2^{\oplus}(f)}. 
  \end{align*}
  The second inequality follows from the fact that each $T$
  is a deterministic parity-decision tree of depth at most $R_2^\oplus(f)$,
  hence it easily follows that the spectral norm of
  the Fourier coefficients of $T$ can be upper bounded by
  $2^{R_2^\oplus(f)}$ (for a proof of this, see~\cite{BOH90}).
  Rewriting the last inequality, we have
  \(\max_{S\subseteq [n]}|\fc{f}{S}| \geq \frac{1}{2^{R_2^\oplus(f) + \log 3}}\),
  which gives us the third part of the theorem.
\end{proof}
Using a well-known fact that upper bounds R\'{e}nyi entropy of order
$1+\delta$ (for every $\delta>0$) by a constant times the min-entropy
of $\left\{\hat{f}(S)^2\right\}$, we deduce the following corollary.
Since this fact works for all $\delta>0$, it is tempting to say that
we can improve the bounds in Theorem~\ref{thm:minentropyupperb} from
$\ent_{\infty}(\hat{f}^2)$ to $\ent(\hat{f}^2)$, but this relation between
the R\'{e}nyi entropies breaks down for $\delta=0$. 
\begin{corollary}\label{coro:renyientropyupperb}
  Let $f\colon\pmset{n}\rightarrow \pmset{}$, and $\delta>0$. Then,
  \begin{enumerate}[label=(\roman*)]
  \item For every $\varepsilon \geq 0$,
    \(\ent_{1+\delta}(\hat{f}^2) \leq 2\left(1+\frac{1}{\delta}\right)\cdot \log \left(\|\widehat{f}\|_{1,\varepsilon}/(1-\varepsilon)\right)\).
  \item \(\ent_{1+\delta}(\hat{f}^2)\leq  2\left(1+\frac{1}{\delta}\right)\cdot \certificate_{\min}^{\oplus}(f)\). 
  \item \(\ent_{1+\delta}(\hat{f}^2)\leq  2(1+\log 3)\left(1+\frac{1}{\delta}\right)\cdot R_2^\oplus(f)\). 
  \end{enumerate}
\end{corollary}
\begin{proof}
  Use the fact that for any distribution $P$,\;
  \(\ent_{1+\delta}(P) \leq \left(1+\frac{1}{\delta}\right)\ent_{\infty}(P)\).
  Indeed, it is easily seen from the definition of R\'{e}nyi entropy
  that \(-\frac{1}{\delta}\log \left(\sum_j p_j^{1+\delta}\right) \leq -\frac{1+\delta}{\delta}\log \left(\max_j\{ p_j\}\right)\),
  and thus the fact follows. We remark that a tighter analysis of
  the R\'{e}nyi entropy can be used to improve the constants.
\end{proof}
As a corollary to Theorem~\ref{thm:minentropyupperb}~\ref{min-vs-cert},
we now establish the FMEI conjecture for read-$k$ $\dnf$s, for constant $k$.
A Boolean function is said to belong to the class of read-$k$ $\dnf$
if it can be expressed as a
$\dnf$ such that every variable (negated or un-negated) appears in at most
$k$ terms. We note that, independently, Shalev~\cite{shalev2018:minentropy}
showed, among other things, that FMEI holds for ``regular'' read-$k$ $\dnf$s.
However, we show it for the general class of read-$k$ $\dnf$s.
We remark that this improvement crucially uses our \emph{sharper} bound
of $\certificate_{\min}$ on the min-entropy of $\left\{\hat{f}(S)^2\right\}_S$. 

We will need the well-known KKL theorem which we state below.
\begin{theorem}[\cite{KKL88,FK96}]
  \label{thm:KKL}
  There exists a universal constant $c > 0$ such that for every $f \colon \pmset{n} \to \pmset{}$, we have
  \begin{align*}
    \Inf(f) \geq c \cdot \Var(f) \cdot \log \frac{1}{\max_i\Inf_i(f)}.
  \end{align*}
\end{theorem}
The next lemma establishes a lower bound of \emph{minimum} certificate size
on the total influence of constant-read $\dnf$.
A similar argument appears in Shalev~\cite{shalev2018:minentropy} too. 
\begin{lemma}
  \label{lem:dnf-inf-lb}
  There exists a universal constant $c>0$ such that for all $f \colon \pmset{n}\rightarrow \pmset{}$ that can be expressed as a read-$k$ $\dnf$, we have 
  \begin{align*}
    \Inf(f) \geq c\cdot \Var(f)\cdot \left(\certificate_{\min}(f) -1 - \log k\right).
  \end{align*}
\end{lemma}
\begin{proof}
  Suppose $f$ is a $t$-term read-$k$ $\dnf$,
  then $f$ can be written as $f(x) = \bigvee_{j=1}^t T_j(x)$,
  where $T_j(x)$ is a term. Recall,
  \(\Inf_i(f) = \Pr_x[f(x) \neq f(x^{(i)})]\).
  Using the fact that $f$ is a $\dnf$ we upper bound
  the $\Inf_i(f)$ as follows,
  \begin{align*}
    \Inf_i(f) = \Pr_x\left[f(x) \neq f(x^{(i)})\right] \leq \sum_{j=1}^t \Pr_x\left[T_j(x) \neq T_j(x^{(i)})\right].
  \end{align*}
  Clearly when $T_j$ is not defined over a variable $x_i$,
  $\Pr_x[T_j(x) \neq T_j(x^{(i)})] =0$, and otherwise it equals
  $\frac{1}{2^{|T_j|-1}}$ because all other literals must be
  set to true in order to satisfy that term. Therefore, we have
  \begin{align*}
    \Inf_i(f) \leq \sum_{j=1}^t \Pr_x\left[T_j(x) \neq T_j(x^{(i)})\right] =\sum_{\substack{j=1\\x_i  \text{ appears in }T_j}}^t \Pr_x\left[T_j(x) \neq T_j(x^{(i)})\right] \leq k2^{-(\certificate_{\min}(f)-1)}.
  \end{align*}
  The second inequality follows because a variable appears
  in at most $k$ terms, and $|T_j| \geq \certificate_{\min}(f)$ for all $j$.
  Now using the KKL theorem (Theorem~\ref{thm:KKL}), we obtain 
  \begin{align*}
    \Inf(f) \geq c \cdot \Var(f)\cdot \log \frac{1}{\max_i\Inf_i(f)} \geq c\cdot\Var(f)\cdot (\certificate_{\min}(f) - 1 -\log k).
  \end{align*}
  This concludes the proof of the lemma.
\end{proof}     
We now use Lemma~\ref{lem:dnf-inf-lb} and
Theorem~\ref{thm:minentropyupperb}~\ref{min-vs-cert}
to show that FMEI conjecture holds for read-$k$ $\dnf$. 
\begin{theorem}[Restatement of Theorem~\ref{thm:meic-dnf-intro}]
  \label{thm:meic-dnf}
  Let $f \colon\pmset{n}\rightarrow \pmset{}$ be a read-$k$ $\dnf$. Then,
  \begin{align*}
    \ent_{\infty}(\hat{f}^2) \leq C \cdot \Inf(f),\; \text{ where } C = O(\log k).
  \end{align*}
\end{theorem}   
\begin{proof}
  We consider two cases based on whether
  $\Var(f)$ is ``small'' $(< 1/2)$ or ``large'' $(\geq 1/2)$. 
  \paragraph{Case 1} $\Var(f) < 1/2$.
  Recall, $\Var(f) = 1 -\fcsq{f}{\emptyset}$. Therefore, we have 
  \begin{align*}
    \ent_{\infty}(\hat{f}^2)\leq  \log \frac{1}{\fcsq{f}{\emptyset}} = \log \frac{1}{1-\Var(f)} \leq (\log e)\frac{\Var(f)}{1-\Var(f)} \leq (2\log e) \Var(f) \leq (2\log e) \Inf(f).
  \end{align*}
  The second inequality uses the fact, for $x \in (0,1)$,
  $\log\frac{1}{1-x} \leq (\log e)\frac{x}{1-x}$. 
  \paragraph{Case 2} $\Var(f) \geq 1/2$.
  Using Theorem~\ref{thm:minentropyupperb}~\ref{min-vs-cert}
  and Lemma~\ref{lem:dnf-inf-lb} we bound the min-entropy as follows, 
  \begin{align*}
    \ent_{\infty}(\hat{f}^2) \leq 2 \cdot \certificate_{\min}(f) \leq (2/c)\cdot \frac{\Inf(f)}{\Var(f)}+ 2(1+\log k),
  \end{align*} where $c$ is a universal constant. 
  Since $\Inf(f) \geq \Var(f) \geq 1/2$,
  we further bound the last term in the
  above inequality to obtain
  \begin{align*}
    \ent_{\infty}(\hat{f}^2) \leq ((4/c)+ 4(1+\log k)) \cdot \Inf(f).
  \end{align*}
  This completes the proof of the theorem.
\end{proof}

\section{Implications of the FEI conjecture}
\label{sec:unbalancinglights}
The FEI and FMEI conjecture seem to impose a strong constraint
on the Fourier spectrum of a Boolean function.
For example, the FMEI conjecture (if true) would show
the existence of a large Fourier coefficient in
the spectrum of every Boolean function that has small average sensitivity.
In the introduction we saw that the FEI conjecture implies 
the existence of a Fourier-sparse polynomial $p$
that approximates a Boolean function $f$ in $\ell_2$-distance,
i.e., $\E_x[(f(x)-p(x))^2]$ is small.
(A particular case being Mansour's conjecture,
which is also a consequence of the FEI conjecture.)
In this section we discuss one implication of the FEI conjecture
relating to the structure of polynomials
that approximate Boolean functions in the $\ell_\infty$-distance,
i.e, $|p(x)-f(x)|$ is small for every $x\in \pmset{n}$.
In particular, we consider the question
``Do polynomials approximating a Boolean function in $\ell_\infty$-distance satisfy some property?''. In this direction,
we make progress by showing that if the FEI conjecture is true,
then we can rule out polynomials with ``large'' Fourier sparsity
from representing or approximating Boolean functions.
In this section, we also consider a class of polynomials
and show that no polynomial in that class can
$1/3$-approximate a Boolean function
(without assuming that the FEI conjecture is true).
\begin{definition}
  An $n$-variate multilinear polynomial $p$ is said to be
  a \emph{flat polynomial} if all  its \emph{non-zero}
  coefficients have the same magnitude.
\end{definition}
\begin{lemma}
  \label{lem:flatpoly}
  Let $\varepsilon \in (0,1/2)$ be a constant, and
  suppose the FEI conjecture is true.
  Let $p$ be a flat polynomial with degree $d$ and
  sparsity $T= 2^{\omega(d)}$.
  Then $p$ cannot $\varepsilon$-approximate any Boolean function.
\end{lemma}
\begin{proof}
  In order to prove the lemma, we crucially use the following claim.
  \begin{claim}\label{cl:entflat}
    Let $p$ be a flat polynomial with sparsity $T$ and, further,
    suppose $p$ $\varepsilon$-approximates a Boolean function $f$. Then, 
    \begin{align}\label{eq:entflat}
      \ent(\hat{f}^2) \geq \Omega(\log T), 
    \end{align}
    where the constant in $\Omega(\cdot)$ depends on $\varepsilon$. 
  \end{claim}
  We assume this claim and conclude the proof of the lemma.
  By contradiction, let us assume that $p$
  $\varepsilon$-approximates a Boolean function $f$,
  so $\deg_{\varepsilon}(f) \leq d$.
  Assuming the FEI conjecture is true,
  we have $\ent(\hat{f}^2) = O(\Inf(f))$.
  Furthermore, using a result of Shi~\cite{shi:approxdeg},
  we have $\Inf(f)\leq O(\deg_{\varepsilon}(f))$ for every $f$
  and constant $\varepsilon$.
  So, we have $\ent(\hat{f}^2) \leq  O(\deg_{\varepsilon}(f))\leq O(d)$.
  Using Claim~\ref{cl:entflat}, it follows that
  \begin{align*}
    \Omega(\log T)\leq \ent(\hat{f}^2) \leq O(d).
  \end{align*}
  But this upper bound of $T=2^{O(d)}$ contradicts
  the assumption on $T$ in the statement of the lemma.
  Hence, we conclude that  $p$ cannot $\varepsilon$-approximate
  any Boolean function.
\end{proof}
We now prove the claim. 
\begin{proof}[Proof of Claim~\ref{cl:entflat}]
  Suppose $p$ $\varepsilon$-approximates a Boolean function $f$ and
  that all non-zero coefficients
  of $p$ have magnitude $\frac{\alpha}{\sqrt{T}}$ for some
  $\alpha\in [(1-\varepsilon),(1+\varepsilon)]$.
  Such an $\alpha$ exists because,
  by Parseval's identity (Fact~\ref{fact:plancherel}), we have
  \begin{align*}
    \sum_S \fcsq{p}{S}=\E_x\left[p(x)^2\right] \in \left[(1-\varepsilon)^2,(1+\varepsilon)^2\right],
  \end{align*}
  where the inclusion assumes $p$ $\varepsilon$-approximates $f$. 
  Now consider the following set $\mathcal{A}$ of ``large''
  Fourier coefficients,
    \begin{align}
      \label{eq:defnofA}
      \mathcal{A} \coloneqq \left\{ S\ :\ |\fc{f}{S}|\geq \frac{2 \alpha }{\sqrt{T}}\right\}.
    \end{align}
    For every $ S\in \mathcal{A}$,
    we have $|\fc{f}{S} - \fc{p}{S}| \geq |\fc{f}{S}|/2$ and hence
    \begin{align}
      \label{eq:uppeboundonfcinA}
      \sum_{S\in \mathcal{A}} \fcsq{f}{S} \leq 4\sum_{S\in \mathcal A} \left(\fc{f}{S} - \fc{p}{S}\right)^2\leq 4\sum_{S\subseteq [n]} \left(\fc{f}{S} - \fc{p}{S}\right)^2 =4\E_x\left[(f(x)-p(x))^2\right] \leq 4 \varepsilon^2,
    \end{align}
    where the equality uses Parseval's identity (Fact~\ref{fact:plancherel})
    and the last inequality uses that $p$ $\varepsilon$-approximates~$f$.
    From Eq.~\eqref{eq:uppeboundonfcinA},
    we have \(\sum_{S\not\in \mathcal{A}} \fcsq{f}{S} \geq 1-4\varepsilon^2\).
    This gives us our desired lower bound on the Fourier entropy of $f$,
    \begin{align*}
      \ent(\hat{f}^2) & =\sum_{S\subseteq [n]} \fcsq{f}{S} \log \frac{1}{\fcsq{f}{S}} \\
      & \geq \sum_{S \not\in \mathcal{A}} \fcsq{f}{S} \log \frac{1}{\fcsq{f}{S}} \geq \sum_{S\not\in \mathcal{A}} \fcsq{f}{S} \cdot \log\Big(\frac{T}{4\alpha^2}\Big)\\
      &\geq (1-4\varepsilon^2) \log\Big(\frac{T}{4\alpha^2}\Big)
      \geq (1-4\varepsilon^2)\log\frac{T}{4(1+\varepsilon)^2}=\Omega(\log T).
    \end{align*}
    The second inequality follows by the definition of
    $\mathcal A$ (Eq.~\eqref{eq:defnofA}),  
    the third inequality by negation of the inequalities
    in Eq.~\eqref{eq:uppeboundonfcinA},
    the last inequality holds because $\alpha \leq (1+\varepsilon)$
    by the definition of $\alpha$,
    and the last equality assumes $\varepsilon \in (0, 1/2)$.
  \end{proof}
Since we are still unable to resolve the FEI conjecture,
an interesting intermediate question would be to
\emph{unconditionally} prove the following conjecture.
\begin{conjecture}
  \label{conj:unbalancing}
  No flat polynomial of degree $d$ and sparsity $2^{\omega(d)}$
  can $\varepsilon$-approximate a Boolean function.
\end{conjecture}
Although we have not been able to resolve this conjecture,
we now discuss some partial progress towards resolving it.
Additionally, we give an intriguing connection between this conjecture
and the Bohnenblust-Hille inequality.

\paragraph{Partial progress towards resolving Conjecture~\ref{conj:unbalancing}}
As a first step towards proving the conjecture we observe that a
\emph{weaker} version of the conjecture holds when further restricting
the approximating polynomials to be \emph{homogeneous}.
A polynomial is said to be \emph{homogeneous} if every monomial in
the polynomial has the same degree. More formally, we observe the following.
\begin{prop}
  \label{prop:hom-conj}
  No flat homogeneous polynomial of degree $d$ and sparsity
  $2^{\omega(d\log d)}$ can $1/3$-approximate a Boolean function. 
\end{prop}
\begin{proof}
  Let $p$ be a flat homogeneous polynomial of degree $d$ that
  $1/3$-approximates a Boolean function $f$.
  Further, suppose the sparsity of $p$ is $T$.
  This implies $|\fc{p}{S}| =\frac{\alpha}{\sqrt{T}}$
  for some $\alpha\in [2/3,4/3]$. Then, we have
  \begin{align}
    \label{eq:diff-l1-norm-ub}
    \sum_{S \colon |S| =d}|\fc{p}{S} - \fc{f}{S}| \leq \left(\sum_{S \colon |S| =d}|\fc{p}{S}-\fc{f}{S}|^2\right)^{1/2} \cdot \sqrt{T} \leq \frac{1}{3} \sqrt{T},
  \end{align}
  where the first inequality is Cauchy-Schwarz and the second
  uses $|p(x) -f(x)| \leq 1/3$ for all $x$. On the other hand,
  \begin{align}
    \label{eq:diff-l1-norm-lb}
    \sum_{S \colon |S| =d}|\fc{p}{S} - \fc{f}{S}| \geq \sum_{S \colon |S| =d}\left(|\fc{p}{S}| - |\fc{f}{S}|\right) & = \sum_{S \colon |S| =d}|\fc{p}{S}| - \sum_{S \colon |S|=d}|\fc{f}{S}| \nonumber \\
    & \geq \frac{2}{3}\sqrt{T} - \sum_{S \colon |S|=d}|\fc{f}{S}|,
  \end{align}
  where the first inequality uses the reverse-triangle inequality and
  the last uses the lower bound on $|\fc{p}{S}|$.
  Combining Eq.~\eqref{eq:diff-l1-norm-ub} and \eqref{eq:diff-l1-norm-lb} 
  we obtain
  \begin{align}
    \label{eq:lb-l1-norm-f}
    \sum_{S \colon |S|=d}|\fc{f}{S}| \geq \frac{1}{3}\sqrt{T}. 
  \end{align}
  However, from results of Tal (see \cite[Claim~2.13]{Tal14} and \cite[Lemmas~29 and~34]{tal:fourierconcentration}) we know
  \begin{align*}
    \sum_{S \colon |S| =d} |\fc{f}{S}|\leq O(d)^{d}.
  \end{align*}
  Thus, from Eq.~\eqref{eq:lb-l1-norm-f} we obtain
  \begin{align*}
    T \leq 2^{O(d\log d)}.
  \end{align*}
\end{proof}

Conjecture~\ref{conj:unbalancing} asks if a similar result also holds
when $T=2^{\omega(d)}$. 
We now consider a restricted class of polynomials called
\emph{block-multilinear polynomials}.
An $n$-variate polynomial is said to be \emph{block-multilinear}
if the input variables can be \emph{partitioned} into disjoint blocks
$A_1, \ldots, A_d \subseteq \{x_1,\ldots,x_n\}$
such that every monomial in the
polynomial has \emph{at most} one variable from each block.
Without loss of generality, we will assume that each block is of the same size. 
In other words, a block-multilinear polynomial
$p\colon(\reals^n)^d\rightarrow \reals$ can be written as
\begin{align}
\label{eq:defnofp}
p(x^1,\ldots,x^d)= \sum_{S\subseteq [nd] \colon \forall j \in [d],\; |S \cap x^j| \leq 1}\fc{p}{S} \prod_{l \in S}x_l,
\end{align}
where $\fc{p}{S} \in \reals$ for every $S \subseteq [nd]$. 
Note that this is the standard Fourier decomposition of $p$
if $x^i\in \pmset{n}$ for every $i \in [d]$.
Clearly such a block-multilinear polynomial has degree at most $d$
and sparsity at most $(n+1)^d$.
Such polynomials have found applications in
quantum computing~\cite{aaronsonambainis:forrelation,montanaro:hypercontractivity}, classical and quantum XOR games~\cite{briet:XOR},
polynomial decoupling~\cite{Odonnell:decoupling} and
in functional analysis which we discuss later.
Our main contribution in this section is that
we show a positive answer to Conjecture~\ref{conj:unbalancing}
for the class of flat block-multilinear polynomials.
\begin{theorem}[Restatement of Theorem~\ref{thm:unbalancing2-intro}]
  \label{thm:unbalancing2}
  If $p$ is an $n$-variate flat block-multilinear polynomial with degree $d$
  and sparsity $2^{\omega(d)}$,
  then $p$ cannot $1/3$-approximate a Boolean function.
\end{theorem}
We defer the proof of this theorem to the next subsection.
We now continue with the relevance of block-multilinear polynomials
and the theorem above to functional analysis literature.

\paragraph{Relation between Theorem~\ref{thm:unbalancing2} and the Bohnenblust-Hille (BH) inequality}
Consider a homogeneous block-multilinear polynomial
$p\colon(\reals^{n})^d\rightarrow \reals$. 
One way to show that $p$ is not $1/3$-close to a Boolean function
would be to show that there exists $x'$ such that $|p(x')|$
is greater than $4/3$. Understanding if such an $x'$ exists
for such a polynomial $p$ can be cast as the following maximization~problem
\begin{align}
  \label{eq:maxp}
  \|p\|:=\max_{x^1,\ldots,x^d \in [-1,1]^n} \Big|\sum_{i_1,\ldots,i_d=1}^n \widehat{p}_{i_1,\ldots,i_d} x^1_{i_1}\cdots x^d_{i_d}\Big|.
\end{align}
In order to develop an intuition for the maximization problem,
consider the case $d=2$ and furthermore suppose
$\widehat{p}_{i_1,i_2}\in \pmset{}$ and $x^i\in \pmset{n}$,
then giving a lower bound on $\|p\|$ is well-known in computer science
as the so-called \emph{unbalancing lights} problem
(see~\cite[Section~2.5]{alon&spencer-book},
where they show the existence of sign vectors such that
$\|p\|\geq \sqrt{n}$).
For larger $d$ and arbitrary $p$,
showing lower bounds on $\|p\|$ in Eq.~\eqref{eq:maxp}
has been extensively studied in the functional analysis literature and
is sometimes referred to as the ``generalized unbalancing lights problem''.

The first paper giving a lower bound to Eq.~\eqref{eq:maxp}
was by Bohnenblust and Hille~\cite{bohnenblust:unbalancing} in 1931.
They gave a lower bound on the injective tensor norm of
degree-$d$ multilinear forms, which in our context translates
to a lower bound on $\|p\|$. To be precise, their result
states the following$\colon$  for every $n$, $d$ there exists
a constant $C_d\geq 1$ such that, for every degree-$d$ homogeneous 
block-multilinear polynomial 
$p\colon(\reals^n)^d\rightarrow \reals$, we  have
\begin{align}
  \label{eq:BHinequality}
  \Big(\sum_{i_1,\ldots,i_d=1}^n |\widehat{p}_{i_1,\ldots,i_d}|^{\frac{2d}{d+1}}\Big)^{\frac{d+1}{2d}}\leq C_d \cdot  \max_{x^1,\ldots,x^d\in [-1,1]^n} |p(x^1,\ldots,x^d)|.
\end{align}
The result of~\cite{bohnenblust:unbalancing} showed that  it suffices to pick
\(C_d=d^{\frac{d+1}{2d}}2^{\frac{d-1}{2}}\)
in order to satisfy Eq.~\eqref{eq:BHinequality}.
In addition, their bound on $C_d$ recovers the well-known
Littlewood's $4/3$ inequality~\cite{littlewood:inequality} when $d=2$.
Since their seminal work, a lot of  research in
the functional analysis literature (for the last 80 years!)
has been in finding the optimal BH-\emph{constants}, i.e.,
the smallest $C_d$ for which Eq.~\eqref{eq:BHinequality} holds.
We cite a few results~\cite{montanaro:hypercontractivity, albuquerque:unbalancing, pellergrino:sharpconstant,defant:unbalancing, defant1:unbalancing,  BPSS14},
referring the reader to the references within these papers for more. 
It is a long-standing open question
whether $C_d$ is a universal constant
(there has been recent work~\cite{pellergrino:sharpconstant}
giving numerical evidence that this is the case).

After a series of works, Pellegrino and Seoane-Sep\'{u}lveda
showed~\cite{pellegrino:newupperbound} that it suffices to pick
$C_d=\text{poly}(d)$. As far as we are aware,
the best upper bound on $C_d$ was shown by Bayart et al.~\cite{BPSS14}
as $C_d=O(d^{0.365})$. 
We also point the interested reader
to~\cite[Theorem~17]{montanaro:hypercontractivity} for a suboptimal,
yet elegant proof that shows that it suffices
to pick $C_d=O(d^{1.45})$ in order to satisfy Eq.~\eqref{eq:BHinequality}.

We now discuss the relevance of the BH-inequality
to a special case of Theorem~\ref{thm:unbalancing2}
where the approximating polynomial is also assumed to be homogeneous.
Consider an arbitrary flat homogeneous block-multilinear polynomial
$p\colon(\pmset{n})^d\rightarrow \reals$ with degree $d$ and sparsity $T$
that $1/3$-approximates a Boolean function.
That is, every non-zero Fourier coefficient
$|\widehat{p}_{i_1,\ldots,i_d}|$ equals $\alpha/\sqrt{T}$
for some  $\alpha \in [2/3,4/3]$.
Then using Eq.~\eqref{eq:BHinequality}, we get
\begin{align*}
  \Big(\Big(\frac{2/3}{\sqrt{T}}\Big)^{\frac{2d}{d+1}} \cdot
  T\Big)^{\frac{d+1}{2d}}\leq C_d\cdot \|p\| \leq (4/3)\cdot C_d.
\end{align*}
With further simplification, we have
\begin{align*}
  T \leq 2^{2d}\cdot C_d^{2d}.
\end{align*}
Using the result of Bayart et al.~\cite{BPSS14},
$C_d=O(d^{0.365})$, we get that the sparsity $T \leq 2^{O(d\log d)}$.
However, from the FEI conjecture (cf.~Claim~\ref{cl:entflat})
it follows that $T \leq 2^{O(d)}$. Thus, using the current best bound
on $C_d$ we cannot even conclude the special case of
Theorem~\ref{thm:unbalancing2}. 
However, if the long-standing open question of $C_d$
being a universal constant were true, then this special case 
follows. But proving $C_d$ a universal constant
seems to be a very hard problem. Nevertheless, in the next section,
we will prove Theorem~\ref{thm:unbalancing2}, 
while circumventing the barrier of improving the upper bound on $C_d$.

We remark that Theorem~\ref{thm:unbalancing2} does not say
anything about the BH-inequality, since Theorem~\ref{thm:unbalancing2}
states that a flat block-multilinear polynomial
either takes a value in the range $(-2/3,2/3)$ or
takes a value of magnitude more that $4/3$,
while the BH-inequality states that such a polynomial definitely
takes a value of high magnitude ($> 4/3$)
on at least one input.

\subsection{Proof of Theorem~\ref{thm:unbalancing2}}
We restate the theorem for convenience.
\paragraph{Theorem~\ref{thm:unbalancing2} (restated)} \emph{If $p$ is an $n$-variate flat block-multilinear polynomial with degree $d$ and sparsity $2^{\omega(d)}$, then $p$ cannot $1/3$-approximate a Boolean function.}

Without loss of generality, let us assume $d$ divides $n$.
Let $\{A_1,\ldots,A_d\}$ be a partition of $\{x_1,\ldots,x_n\}$
and for simplicity suppose $A_1 =\{x_1, \ldots, x_k\}$ for $k =n/d$.
Let $B=\{x_1,\ldots,x_n\}\setminus A_1$
denote the set of remaining variables.
Then the monomials of $p$ can be divided into
those containing variables in $A_1$ and
those independent of variables in $A_1$. So we can write $p$ as follows
\begin{align}
  \label{eq:rewritep}
  p(x) = q_0(x_B)+ \sum_{i=1}^kx_iq_i(x_B),
\end{align}
where $q_0,\ldots,q_k$ are polynomials of degree at most $d-1$.
From here on, for notational simplicity we simply rewrite $x_B$ as $z$.
\begin{lemma}
  \label{lem:properties-of-q}
  Let $p$ be a degree-$d$ block-multilinear polynomial that
  $1/3$-approximates a Boolean function
  $f \colon\pmset{n}\rightarrow \pmset{}$.
  Further,  let $q_0,\ldots,q_k$ be as defined in Eq.~\eqref{eq:rewritep}. Then,
  \begin{enumerate}[label=(\Roman*)]
  \item\label{item:1} For every $z \in\pmset{n-k}$,
    we have $q_i(z) \in [-4/3,-2/3]\cup[-1/3,1/3]\cup[2/3,4/3]$,
    i.e., $q_i$ is a $\frac{1}{3}$-approximation to
    a $\{-1,0,1\}$-valued function.
  \item\label{item:2} For every $z \in \pmset{n-k}$,
    there exists a unique $j \in \{0,\ldots, k\}$ which satisfies
    \begin{align*}
      |q_j(z)| \geq 2/3\quad \text{ and } \quad   \sum_{i \neq j} |q_i(z)| \leq 1/3.
    \end{align*}
  \end{enumerate}
\end{lemma}
\begin{proof}
  The proof of the first part is fairly straightforward,
  while the second part requires some calculations.
  \begin{proof}[Proof of~\ref{item:1}]
  We first rewrite $q_1,\ldots,q_k$ as follows:
  \begin{align*}
    q_i(z)=\frac{p(x_1,\ldots, x_{i-1},1,x_{i+1},\ldots,x_k,z) - p(x_1,\ldots,x_{i-1},-1,x_{i+1},\ldots ,x_k,z)}{2},
  \end{align*}
  for $i \in [k]$ and similarly we rewrite $q_0$ as
  \begin{align*}
    q_0(z)=\frac{p(x_1,\ldots,x_k,z) + p(-x_1,\ldots ,-x_k,z)}{2}.
  \end{align*}
  Now the first part follows from the fact that $p$ is
  a $1/3$-approximation to the Boolean function $f$.
  So the $q_i$s are $1/3$-approximation to a $\{-1,0,1\}$-valued function.
  \end{proof}
  \begin{proof}[Proof of~\ref{item:2}]
    Fix $z \in \pmset{n-k}$. First observe that for every
    $x_1,\ldots,x_k\in \pmset{}$, we have
    \begin{align}
      \label{eq:intervalofp}
      p(x_1,\ldots,x_k,z) \in \underbrace{\left[q_0(z)-\sum_{i=1}^k|q_i(z)|,~ q_0(z)+\sum_{i=1}^k|q_i(z)|\right]}_{:=\mathcal{I}(z)}.
    \end{align}
    Furthermore, there exists a choice of
    $x_1,\ldots ,x_k \in \pmset{}$
    such that $p$ evaluates to either one of the end points
    in the interval $\mathcal{I}(z)$.
    Moreover $\mathcal{I}(z)$ satisfies
    \emph{exactly} one of the following containments:
  \begin{enumerate}[label=(\alph*)]
  \item $\mathcal{I}(z) \subseteq [2/3,4/3]$, or
  \item $\mathcal{I}(z) \subseteq [-4/3,-2/3]$, or
  \item $\mathcal{I}(z)$ intersects both the intervals
    $[-4/3,-2/3]$ and $[2/3,4/3]$.
  \end{enumerate}
  We now prove that in all the three cases,
  there exists $j\in \{0,\ldots,k\}$ such that
  $|q_j(z)| \geq 2/3$ and  $\sum_{i\neq j}|q_i(z)| \leq 1/3$,
  hence proving the lemma statement. 
  \paragraph{Cases (a) and (b)}
  The proofs for Case~(a) and (b) are exactly the same,
  so we prove the lemma assuming
  $\mathcal{I}(z) \subseteq [2/3,4/3]$. In this case, observe that
  \(q_0(z)-\sum_{i=1}^k|q_i(z)| \geq 2/3\) and
  \(q_0(z)+\sum_{i=1}^k|q_i(z)| \leq 4/3\).
  Clearly this implies $q_0(z) \geq 2/3$.
  Furthermore, from both the inequalities it follows that
  \begin{align*}
    \sum_{i=1}^k|q_i(z)| \leq \min\left\{q_0(z) -2/3, 4/3 - q_0(z)\right\}.
  \end{align*}
  Using the fact that $q_0(z) \in [2/3, 4/3]$,
  it follows that $\sum_{i=1}^k|q_i(z)| \leq 1/3$,
  which concludes the proof for Case~(a).
  \paragraph{Case (c)}
  Recall the range of $p(x_1,\ldots ,x_k,z)$
  from Eq.~\eqref{eq:intervalofp}. 
  Let $a,b \in \pmset{k}$ be such that
  \begin{align}
    \label{eq:defnofpab}
    p(a,z) = q_0(z)-\sum_{i=1}^k|q_i(z)| \quad \text{ and } p(b,z) = q_0(z)+\sum_{i=1}^k|q_i(z)|.
  \end{align}
  It is not hard to see that $a = -b$.
  Consider a path from $a$ to $b$ on the Boolean hypercube.
  Since $p(a,z) \in [-4/3,-2/3]$ and $p(b,z) \in [2/3,4/3]$,
  there exists an edge\footnote{An \emph{edge} in the \emph{direction} of $j\in [k]$ on the Boolean hypercube refers to a tuple $(w,w^{(j)})$ where $w,w^{(j)}\in \pmset{k}$ and $w^{(j)}$ is the bit string obtained by flipping the sign of the $j$th bit of $w$.}
  on this path such that the value of $p$ on this edge
  jumps from the interval $[-4/3,-2/3]$ to $[2/3,4/3]$.
  For now the existence of such an edge is sufficient,
  below we explicitly construct such an edge.
  Let us denote the direction of this edge by $j \in [k]$.
  We now claim that $|q_j(z)| \geq 2/3$ and
  $\sum_{i \neq j} |q_i(z)| \leq 1/3$.

  The first claim $|q_j(z)| \geq 2/3$ follows immediately,
  because the change in $p$ on this edge is
  $2\cdot |q_j(z)|$ (by Eq.~\eqref{eq:defnofpab}), and
  the change of value of $p$ on this edge is $\geq 2\cdot (2/3)$.
  Thus we have $|q_j(z)| \geq 2/3$.
  
  We now prove $\sum_{i \neq j} |q_i(z)| \leq 1/3$ by considering
  two cases based on whether $\mathsf{sign}(q_0(z))$ is positive or negative.
  Since the proofs for both cases are similar,
  for simplicity we prove it assuming $\mathsf{sign}(q_0(z))$ is positive.
  We now define the edge between $a,b$ where the value of
  $p$ jumps from the interval $[-4/3,-2/3]$ to $[2/3,4/3]$.
  Consider the edge in the $j$th direction given by setting
  all the variables except $x_j$ as follows$\colon$
  \(x_i = \mathsf{sign}(q_i(z))$ for $i\in[k]\setminus \{j\}\).
  Clearly the value of $p$ only depends on $x_j$ and
  equals \(\sum_{i \neq j} |q_i(z)| + x_jq_j(z)\).
  When $x_j = \mathsf{sign}(q_j(z))$,
  $p$ takes the value \(\sum_{i \neq j} |q_i(z)| + |q_j(z)|\)
  which is in the interval $[2/3, 4/3]$ and
  \begin{align}
    \label{eq:boundonsumqi9/9}
    \sum_{i  \neq j} |q_i(z)| + |q_j(z)| \in [2/3,4/3], \text{ implies } \sum_{i  \neq j} |q_i(z)|\leq 4/3 - |q_j(z)|.
  \end{align}
  Similarly, when $x_j = -\mathsf{sign}(q_j(z))$
  then $p$ takes the value
  \(\sum_{i \neq j} |q_i(z)| - |q_j(z)|\)
  which is in the interval $[-4/3, -2/3]$, and
  \begin{align}
    \label{eq:boundonsumqi9/9-2}
    \sum_{i  \neq j} |q_i(z)| - |q_j(z)| \in [-4/3,-2/3], \text{ implies } \sum_{i  \neq j} |q_i(z)|
    \leq |q_j(z)|-2/3.
  \end{align}	
  From Eq.~\eqref{eq:boundonsumqi9/9} and~\eqref{eq:boundonsumqi9/9-2}
  we have $\sum_{i \neq j} |q_i(z)| \leq 1/3$. 
  The uniqueness of $j$ in each case follows because
  $\sum_{i=0}^k|q_i(z)| \leq 4/3$.
  \end{proof}
\end{proof}
We now show that if $p(x)$ $\frac{1}{3}$-approximates
a Boolean function $f$ then $\deg(f) \leq d$.
\begin{lemma}
  \label{lem:block-multilinear-deg}
  Let $p$ be a degree-$d$ block-multilinear polynomial that
  $\frac{1}{3}$-approximates a Boolean function
  $f \colon\pmset{n}\rightarrow \pmset{}$.
  Then, \(\deg(f) \leq d.\)
\end{lemma}
\begin{proof}
  We will establish the proof by induction on $d$. 
  \paragraph{Base case} We have $d=1$.
  Let $p(x) = a_0 + \sum_{i=1}^na_ix_i$, where $a_i \in \reals$.
  Since $p$ approximates $f$, it follows that
  $\mathsf{sign}(p(x)) = f(x)$ for every $x$. We now express
  $\mathsf{sign}(p(x))$ as a polynomial of degree at most $1$.
  From Lemma~\ref{lem:properties-of-q}~\ref{item:2}, it follows that there
  exists a unique $i \in \{0\}\cup[n]$ such that
  \begin{align*}
    \mathsf{sign}(p(x)) =
    \begin{cases}
      \mathsf{sign}(a_i)\cdot x_i & \text{if } i \in [n], \\
      \mathsf{sign}(a_0) & \mbox{otherwise.}
    \end{cases}
  \end{align*}
  Thus the base case follows.  
  \paragraph{Inductive assumption} Assume the lemma statement for
  polynomials $p$ of degree at most $d-1$. 
  \paragraph{Inductive step} Consider the decomposition of
  $p(x_1,\ldots ,x_n)$
  \begin{align*} p(x) = q_0(z) +\sum_{i=1}^kx_iq_i(z).\end{align*}
  Recall from the proof of Lemma~\ref{lem:properties-of-q}~\ref{item:1},
  that $q_i(z)$ can be expressed as a difference between two
  $(d-1)$-block-multilinear polynomials by fixing $x_1,\ldots, x_k$. That is,
  \begin{align}
    \label{eq:defnofq0}
    \begin{aligned}
      q_0(z) & = \frac{p(1,\ldots,1,z) + p(-1,\ldots ,-1,z)}{2}, \\
      q_i(z) & = \frac{p(1,\ldots, 1,1,1,\ldots,1,z) - p(1,\ldots,1,-1,1,\ldots ,1,z)}{2} \quad \text{ for } i\in [k],
    \end{aligned}
  \end{align}
  where the $-1$ in the expression
  $p(1,\ldots,1,-1,1,\ldots ,1,z)$ is at the $i$-th coordinate.
  Clearly, it follows that $p(-1,\ldots,-1,z)$ and
  $p(1,\ldots,1,z)$ are block-multilinear polynomials
  with $d-1$ parts that $\frac{1}{3}$-approximate
  Boolean functions $f(-1,\ldots ,-1,z)$ and $f(1,\ldots,1,z)$, respectively.
  Therefore, by the inductive assumption,
  both $\deg(f(-1,\ldots,-1,z))$ and $\deg(f(1,\ldots,1,z))$ are
  $\leq d-1$. Additionally the function $\widetilde{q}_0$ defined as
  \begin{align*}
    \widetilde{q}_0(z) = \frac{f(1,\ldots,1,z) + f(-1,\ldots ,-1,z)}{2}
  \end{align*}
  is a $\{-1,0,1\}$-valued function satisfying
  $\deg(\widetilde{q}_0) \leq d-1$.
  Moreover, $q_0$ (as defined in Eq.~\eqref{eq:defnofq0})
  is a  $\frac{1}{3}$-approximation to $\widetilde{q}_0$.
  In a similar fashion, one can define $\widetilde{q_i}$
  for all $i \in [k]$  which is a $\{-1,0,1\}$-valued function
  of degree at most $d-1$ that additionally satisfies that
  $q_i$ (as defined in Eq.~\eqref{eq:defnofq0}) is
  a $\frac{1}{3}$-approximation to $\widetilde{q}_i$.
  Finally consider the polynomial $\widetilde{p}$ defined as follows,
  \begin{align*}
    \widetilde{p}(x) = \widetilde{q}_0(z) +\sum_{i=1}^kx_i\widetilde{q}_i(z).
  \end{align*}
  Firstly note that $\widetilde{p}(x)$ is a polynomial
  that takes value in $\{-1, +1\}$ for all $x\in \{-1, +1\}^n$.
  This is because for all $i$, $q_i(z)$ $\frac{1}{3}$-approximates
  $\widetilde{q}_i(z)$ and so from Lemma~\ref{lem:properties-of-q},
  for any given $z$ there exists exactly one $i$ such that
  $\widetilde{q}_i(z)$ takes a non-zero value. So  for a given $x$
  if $j$ is the unique index from  Lemma~\ref{lem:properties-of-q}
  such that $|q_j(z)| \geq 2/3$,
  then for that $x$ we have $\widetilde{p}(x) = x_j\widetilde{q}_j(z)$.
  
  From this observation we also get
  $|f(x) - \widetilde{p}(x)| \leq 2/3$ for any $x \in \pmset{n}$.
  This is because for a given $x$ if $j$ is the unique index
  from Lemma~\ref{lem:properties-of-q} such that
  $4/3\geq |q_j(z)| \geq 2/3$, then $f(x) = \mathsf{sign}(x_jq_j(z))$
  and $|f(x) - x_jq_j(z)| \leq 1/3$
  (because of  Lemma~\ref{lem:properties-of-q}~\ref{item:2}), 
   $|x_jq_j(z) - x_j\widetilde{q}_j(z)| \leq 1/3$
  (by definition of $\widetilde{q}_j$) and
  $\widetilde{p}(x) = x_j\widetilde{q}_j(z)$.
  Since both $f$ and $\widetilde{p}$ are
  $\{-1, +1\}$-valued functions and 
  $|f(x) - \widetilde{p}(x)| \leq 2/3$,
  it follows that $\widetilde{p}(x) = f(x)$.
	
  By construction, the degree of $\widetilde{p}$ is at most $d$.
  The lemma now follows as $\widetilde{p}(x) = f(x)$.
\end{proof}
Using Claim~\ref{cl:entflat} and Lemma~\ref{lem:block-multilinear-deg},
it now follows that if $p$ is a degree-$d$ flat block-multilinear
polynomial that $1/3$-approximates a Boolean function,
then the sparsity of $p$ is at most $2^{O(d)}$.
This completes the proof of Theorem~\ref{thm:unbalancing2}.

\section{Conclusion}
\label{sec:conclusion}
We gave improved upper bounds on Fourier entropy of
Boolean functions in terms of average unambiguous
(parity)-certificate complexity, and as a corollary
verified the FEI conjecture for functions with
bounded average unambiguous (parity)-certificate complexity.
We established many bounds on Fourier min-entropy
in terms of analytic and combinatorial measures,
namely minimum certificate complexity,
logarithm of the approximate spectral norm and
randomized (parity)-decision tree complexity.
As a corollary to this,
we verified the FMEI conjecture for read-$k$ $\dnf$s.
We also studied structural implications of the FEI conjecture
on approximating polynomials.
In particular, we proved that flat block-multilinear
polynomials of degree $d$ and sparsity $2^{\omega(d)}$
can not approximate Boolean functions.

We now list a few open problems which we believe
are structurally interesting and could lead towards proving
the FEI or FMEI conjecture.
Let $f\colon\pmset{n}\rightarrow \pmset{}$ be a Boolean function.
\begin{enumerate}
\item Does there exist a Fourier coefficient $S\subseteq [n]$
  such that $|\widehat{f}(S)|\geq 2^{-O(\deg_{1/3}(f))}$?
  This would show $\ent_{\infty}(\hat{f}^2)\leq O(\deg_{1/3}(f))$.
\item Can we show $\ent(\hat{f}^2)\leq O(Q(f))$? Or,
  $\ent_{\infty}(\hat{f}^2)\leq O(Q(f))$?
  (where $Q(f)$ is the $1/3$-error quantum query complexity of~$f$,
  which Beals et al.~\cite{bbcmw:polynomialsj} showed
  to be at least $\deg_{1/3}(f)/2$).
\item Does there exist a universal constant $\lambda>0$ such that
  $\ent(\hat{f}^2)\leq \lambda\cdot \min\{\certificate^1(f),\certificate^0(f)\}$? This would resolve Mansour's conjecture.
\end{enumerate}
In an earlier version of this manuscript we suggested that
bounding the logarithm of the approximate spectral norm by
$O(\deg_{1/3}(f))$ or $O(Q(f))$ might be an approach to answer
Question (1) or (2) above. However,
in a very recent work \cite{CCMP19} it is shown that
$\log (\|\widehat{f}\|_{1,\varepsilon})$ could be as large as
$\Omega(Q(f)\cdot \log n)$, thus nullifying the suggested approach.

\begin{acks}
Part of this work was carried out when NS and SC visited CWI, Amsterdam. SA did most of this work while a Postdoc at Center for Theoretical Physics at MIT, PhD student at QuSoft, CWI,
Amsterdam, and a visitor at University of Bristol (partially supported by EPSRC grant EP/L021005/1).  SA thanks Ashley Montanaro for his hospitality.
NS and SC would like to thank Satya Lokam for many helpful discussions on the Fourier entropy-Influence conjecture. SA and SC thank Jop Bri\"{e}t for pointing us to the literature on unbalancing lights and many useful discussions regarding Section~\ref{sec:unbalancinglights}. We also thank Penghui Yao and Avishay Tal for discussions during the course of this project, and Fernando Vieira Costa J\'{u}nior for pointing us to the reference~\cite{BPSS14}. Finally, we thank the anonymous reviewers for many helpful comments that greatly improved the presentation of the paper. An extended abstract of this manuscript appeared in the conference proceedings of STACS 2020 \cite{ACKSW20}. 

\paragraph{Funding}
\subparagraph{Srinivasan Arunachalam:} Work done when at QuSoft, CWI, Amsterdam, supported by ERC Consolidator Grant 615307 QPROGRESS and MIT-IBM Watson AI Lab under the project {\it Machine Learning in Hilbert space}.
\subparagraph{Michal Kouck{\'{y}}:} Partially supported by ERC Consolidator Grant 616787 LBCAD, and GA\v{C}R grant 19-27871X.
\subparagraph{Nitin Saurabh:} Part of the work was done when the author was at IUUK Prague supported by the European Union's Seventh Framework Programme (FP/2007-2013)/ERC Grant Agreement no. 616787, and at Max Planck Institut f\"{u}r Informatik, Saarbr\"{u}cken, Germany.
\subparagraph{Ronald de Wolf:} Partially supported by ERC Consolidator Grant 615307 QPROGRESS (ended Feb 2019) and by NWO under QuantERA project QuantAlgo 680-91-034 and the Quantum Software Consortium.
\end{acks}

\bibliographystyle{ACM-Reference-Format}
\bibliography{entropy}


\end{document}